\documentclass{amsart}

\usepackage[ascii]{inputenc}
\usepackage{amsmath,amsfonts,amssymb,amsthm}
\usepackage{hyperref,graphicx}

\usepackage[dvipsnames]{xcolor}

\usepackage[shortlabels]{enumitem}

\theoremstyle{plain}
\newtheorem{theorem}{Theorem}
\newtheorem{lemma}{Lemma}
\newtheorem{proposition}{Proposition}
\newtheorem{cor}{Corollary}

\theoremstyle{remark}
\newtheorem{remark}{Remark}
\newtheorem*{remark*}{Remark}

\theoremstyle{definition}
\newtheorem*{definition*}{Definition}
\newtheorem{definition}{Definition}
\newtheorem{hypothesis}{Hypothesis}
\newtheorem{example}{Example}
\newtheorem*{problem*}{Problem}

\def\C{\mathbb{C}}
\def\N{\mathbb{N}}
\def\R{\mathbb{R}}
\def\Z{\mathbb{Z}}

\def\P{\mathbb{P}}
\def\k{\kappa}
\def\U{U^{-*}}
\def\H{\mathcal{H}}

\def\W{\mathcal{W}}
\def\E{\mathbb{E}}
\def\T{\mathbb{T}}
\def\supp{{\rm supp}}
\def\u{\upsilon}
\def\s{b}
\def\Div{{\rm div}}
\def\sgn{{\rm sgn}}
\def\xk{\xi_l}

\def\O{\mathcal{O}}
\def\L{Q}
\renewcommand{\epsilon}{\varepsilon}

\def\Ree{\mathrm{Re}}

\newcommand{\norm}[1]{{\left\Vert {#1}\right\Vert}}

\begin{document}
\title[Infinite dimensional anisotropic compressed sensing]{Infinite dimensional compressed sensing from anisotropic measurements and applications to inverse problems in PDE}

\author{Giovanni S. Alberti}
\address{Department of Mathematics, University of Genoa, Via Dodecaneso 35, 16146 Genova, Italy.}
\email{giovanni.alberti@unige.it}

\author{Matteo Santacesaria}
\address{Department of Mathematics, University of Genoa, Via Dodecaneso 35, 16146 Genova, Italy.}
\email{matteo.santacesaria@unige.it}


\date{May 25, 2019}

\begin{abstract}
We consider a compressed sensing problem in which both the measurement and the sparsifying systems are assumed to be frames (not necessarily tight) of the underlying Hilbert space of signals, which may be finite or infinite dimensional. The main result gives explicit bounds on the number of  measurements in order to achieve stable recovery, which depends on the mutual coherence of the two systems. As a simple corollary, we prove the efficiency of nonuniform sampling strategies in cases when the two systems are not incoherent, but only asymptotically incoherent, as with the recovery  of wavelet coefficients from Fourier samples. This general framework finds applications to inverse problems in partial differential equations, where the standard assumptions of compressed sensing are often not satisfied. Several examples are discussed, with a special focus on electrical impedance tomography.
\end{abstract}


\maketitle

\tableofcontents

\section{Introduction}

The recovery of a sparse signal from a small number of samples is the fundamental question of compressed sensing {(CS)}. A signal is called sparse if it can be expressed as a linear combination of a small number of known vectors. The seminal papers \cite{CRT,2006-donoho} have triggered an impressive amount of research in the last decade, from real world applications (MRI, X-ray tomography, etc.) to theoretical generalizations in broader mathematical frameworks \cite{FR}.

In the finite dimensional case, the general {CS} problem can be stated as follows. Given an unknown {sparse} vector $x_0 \in \C^N$ and a measurement operator represented by a matrix $U\in\C^{N\times N}$, we want to reconstruct $x_0$ from samples of the form $(U x_0)_l$, for $l \in \Omega \subseteq \{1,\ldots,N\}$. This is done by solving the convex optimization problem
\begin{equation}\label{def:l1in}
\min_{ x \in \C^N} \|  x \|_{\ell^1} \quad \text{subject to } P_\Omega U x = P_\Omega U x_0,
\end{equation}
where $P_\Omega$ is the projection matrix on the entries indexed by $\Omega$. It is natural to ask under what conditions the solution of the minimization problem \eqref{def:l1in} coincides with $x_0$. These can be formulated as a lower bound on the number of measurements $m = |\Omega|$, which depends on the sparsity of the signal $s = |\supp( x_0)|$, the dimension $N$ of the ambient space, and the matrix $U$. An interesting feature is that the lower bound on $m$ does not guarantee {exact recovery} for all set of indices $\Omega \subseteq \{1,\ldots,N\}$ with $|\Omega|= m$, but only for \textit{most} of them. 

One of the first contributions \cite{CRT} considered the case where $U$ is the discrete Fourier transform: exact recovery is guaranteed with high probability provided that $\Omega \subseteq\{1,\dots,N\}$ is selected uniformly at random with $m\gtrsim s\log N$.
If $U$ is a general unitary transformation, the problem has been addressed for the first time in \cite{2007-candes}, introducing the \textit{coherence} $\mu = \max_{l,j}|U_{lj}|$. In this case, {the bound becomes $m\gtrsim s \mu^2 N \log N$. }

Similar results have been recently obtained in the infinite dimensional setting, where one considers signals belonging to a separable Hilbert space $\H$ and the measurement operator is represented as a bounded linear map $U\colon\H\to\ell^2(\N)$. (Note that $U$ may  be expressed by scalar products with a family of vectors $\{\psi_l\}_l\subseteq\H$, namely $(Uf)_l=\langle f,\psi_l\rangle_\H$.) The sparsity of a signal $f\in\H$ is characterized by the sparsity of $Df$, where  $D\colon\H\to\ell^2(\N)$, $f\mapsto (\langle f,\varphi_j\rangle_\H)_j$, is the analysis operator associated with a family of vectors $\{\varphi_j\}_j\subseteq\H$. The first results in this framework were presented in \cite{AH}, in the case where both $U$ and $D$ are unitary operators, i.e.\ correspond to orthonormal bases; the orthonormality of $\{\psi_l\}_l$ is a standard assumption taken in virtually all works on CS with deterministic measurements. These results were further extended in \cite{2017-adcock-hansen-poon-roman}, introducing the more advanced concepts of asymptotic incoherence, local coherence, and local sparsity. An additional improvement was given in \cite{Poon2015}, which deals with the case where {$\{\varphi_j\}_j$} is a Parseval frame {(see also \cite{2012-liu,2015-krahmer-needell-ward,2016-giryes})}. Allowing the system $\{\varphi_j\}_j$ to be a frame, i.e.\ $D$ is not necessarily invertible, is useful whenever we wish to use a redundant representation to sparsify the signals in $\H$ (e.g.\ redundant wavelets \cite{2005-fowler}, curvelets
\cite{2004-curvelets}, ridgelets \cite{1999-ridgelets} and
shearlets \cite{2005-labate-lim-kut-weiss,2011-kutyniok-lim,2012-kittipoom-kutyniok-lim}). 

In a large number of inverse problems, where one does not have complete freedom in the measurement process, the assumption on $U$ being unitary is not verified, thereby preventing the application of CS to many domains. As a result, two large research areas as inverse problems in partial differential equations (PDE) and sparse signal recovery  have  been almost completely separated so far. The purpose of this paper is to provide a solid foundation that is expected to allow a fruitful interaction between these two domains.

In order to do so, in this work we present a very general CS result that deals with any bounded and injective linear operators $U$ and $D$, defined on any separable complex Hilbert space (finite or infinite dimensional). Equivalently, the families $\{\psi_l\}_l$ and $\{\varphi_j\}_j$ are simply required to be frames of $\H$ (not necessarily tight). Since we do not need the measurement operator $U$ to be unitary, our results cover the case of \textit{anisotropic} measurements. These have  already been studied in the finite dimensional case using random and not deterministic measurements in \cite{KG}. As far as we know, our result is new also in the finite dimensional case. We consider the analysis formulation of the $\ell^1$ optimization problem, in contrast to the synthesis formulation given for simplicity in \eqref{def:l1in}.

Another generalization is related to the sampling strategy. Recently, it has been observed in several works \cite{2015-krahmer-needell-ward,2017-adcock-hansen-poon-roman} that, when precise bounds for the mutual coherence are available, uniform sampling strategies do not give sharp estimates for the minimum number of measurements. Our techniques are also able to cover this case, also known as \textit{structured} sampling, just as a simple corollary of the main result for the uniform sampling. To our knowledge, this is the first infinite dimensional result under asymptotic incoherence assumptions, where there is no need to use multi-level sampling strategies and local coherence. For instance, we analyze the recovery of wavelet coefficients from Fourier samples, and justify the use of a nonuniform sampling scheme corresponding to the so-called log sampling. This represents only a first step, and we believe that many other interesting estimates may be derived as corollaries of the main general result. 

As mentioned above, our main motivation in dealing with the infinite dimensional anisotropic framework comes from inverse problems arising from partial differential equations \cite{2017-isakov}. These inverse problems are intrinsically infinite dimensional, {and often} the measurement operator cannot be chosen as a unitary transformation. Moreover, in order to obtain a solution to these problems, an infinite number of measurements is often needed, even when the signal to be recovered belongs to a known finite dimensional subspace. CS can thus provide a rigorous, explicit and numerically viable way to find solutions to such problems when only a finite number of measurements is available. We explore applications of our main result to the problems of nonuniform Fourier sampling, (linearized) electrical impedance tomography (EIT)  and photoacoustic tomography. In particular, we show for the first time that an electrical conductivity may be recovered from a number of linearized EIT measurements proportional to the sparsity of the signal with respect to a wavelet basis, up to log factors. Many other inverse problems can be tackled with a similar approach and will be the subject of future work. The interaction between CS and partial differential equations is not limited to inverse problems: CS has been recently used for designing numerical methods for solving PDE \cite{bouchot2017multi,rauhut-schwab-2017,brugiapaglia-etal-2018}.

The plan of the paper is the following. In Section~\ref{sec:mainass} we introduce the mathematical framework of  infinite dimensional CS using the language of frames. We define the mutual coherence for general frames as well as the balancing property. The main result is presented in Section~\ref{sec:thm}, which contains also the main corollary about structured sampling and asymptotic incoherence. Section~\ref{sec:apps} is devoted to  applications  to three inverse problems, while Section~\ref{sec:proof} contains the proof of the main result.

\section{Main assumptions}\label{sec:mainass}

Let $\N$ denote the set of all positive natural numbers. Let $\H$ be a separable complex Hilbert space, representing our signal space, {which may be either finite or infinite dimensional}. The problem we study in this paper is the recovery of  an unknown signal $g_0\in\H$ from partial measurements of the form $(\langle g_0,\psi_l\rangle_\H)_l$, under a sparsity assumption on $g_0$ with respect to a suitable family of vectors $\{\varphi_j\}_j$. The main assumption of this paper is the following: these families of vectors are required to be frames  of $\H$ \cite{2000-casazza,2001-Christensen,2016-christensen}.
\begin{hypothesis}\label{hp1}
{Let $L$ and $J$ be two index sets\footnote{{We say that $I\subseteq\N$ is an index set if $I=\N$ or $I=\{1,2,\dots,n\}$ for some $n\in\N$.}}.} Suppose that $\{\psi_l\}_{l\in L}$ and $\{\varphi_j\}_{j\in J}$ are two frames of $\H$ with frame constants $A_U,B_U>0$ and $A_D,B_D>0$, respectively, namely 
\begin{equation*}
A_U\norm{g}_\H^2\le\sum_{l\in L} |\langle g,\psi_l\rangle_\H|^2\le B_U\norm{g}_\H^2,\qquad 
A_D\norm{g}_\H^2\le\sum_{j\in J} |\langle g,\varphi_j\rangle_\H|^2\le B_D\norm{g}_\H^2,
\end{equation*}
for every $g\in\H$.
\end{hypothesis}
The measurements and the sparsity  are expressed by the analysis operators {$U\colon\H\to\ell^2(L)$ and $D\colon\H\to\ell^2(J)$}, defined by
\begin{equation*}
(Ug)_l = \langle g,\psi_l\rangle_\H, \qquad       (Dg)_j = \langle g,\varphi_j\rangle_\H.
\end{equation*}
By construction, the dual operators are given by $U^*e_l =\psi_l$ and $D^*e_j=\varphi_j$, where $\{e_i\}_{i\in I}$ is the canonical basis of $\ell^2(I)$. By Hypothesis~\ref{hp1}, since $\sum_{l} |\langle g,\psi_l\rangle_\H|^2 = \norm{Ug}_2^2$ and $\sum_{j} |\langle g,\varphi_j\rangle_\H|^2 = \norm{Dg}_2^2$, we have that $U$ and $D$ are bounded and the operator norms satisfy
\begin{equation}\label{eq:normUD}
\norm{U}=\norm{U^*}\le \sqrt{B_U},\qquad 
\norm{D}=\norm{D^*}\le \sqrt{B_D}.
\end{equation}

The recovery problem can then be stated as follows: given noisy partial measurements of $Ug_0$, namely $\zeta = P_\Omega Ug_0+\eta$ for some (finite) set $\Omega\subseteq {L}$, recover the signal $g_0\in\H$, under the  assumption that $Dg_0$ is sparse. 
Here we have used the notation $P_\Omega$ for the orthogonal projection onto  $\overline{\mathrm{span}\{e_j : j\in \Omega\}}$ (if $\Omega=\{1, \ldots, N\}$ we simply write $P_N$). 
The classical way to solve this problem is via $\ell^1$ minimization, namely
\begin{equation}\label{eq:l1}
\inf_{ \substack{
 g \in \H \\  D  g \in \ell^1({J})
}} \Vert D  g \Vert_{1} \quad \text{subject to } \norm{P_\Omega U  g - \zeta}_{2}\le\epsilon,
\end{equation}
where $ \norm{\eta}_2\le\epsilon $ is the noise level.

Equivalently, one may adapt a more abstract point of view, starting from a bounded operator $U\colon\H\to \ell^2(L)$ with bounded inverse. It is immediate to verify that $\psi_l=U^*e_l$ gives rise to a frame, as in Hypothesis~\ref{hp1}. The formulation with $U$ allows to consider any linear inverse problem of the form 
\[
U\colon \H\to\ell^2({L}),\qquad Ug=\zeta.
\]
The only requirement is that, with full data, the inverse problem should be uniquely and stably solvable. In particular, any linear invertible operator $U$ may be considered, and not necessarily isometries as in the standard compressed sensing setting (see Section~\ref{sec:apps}).

\begin{remark*}
The formulation given in \eqref{eq:l1} of the $\ell^1$ optimization problem is the \emph{analysis} approach, because of the minimization of $\Vert D  g \Vert_1$, where $D$ is the analysis operator. This is in contrast with the  \emph{synthesis} formulation
\begin{equation}\label{eq:l1-synthesis}
\inf_{ x\in\ell^1({J})} \Vert x \Vert_{1} \quad \text{subject to } \norm{P_\Omega U  D^* x - \zeta}_2 \le\epsilon.
\end{equation}
In general, the two approaches are not equivalent \cite{2007-analysis_vs_synthesis} (see also \cite{genzel2017}). We have decided to work with the analysis approach since, if $D$ yields a redundant representation, there may be multiple minimizers of \eqref{eq:l1-synthesis} that give the same $D^* x$. Thus uniqueness of minimizers is expected to hold for $g = D^* x$ in \eqref{eq:l1} but not for $x$ in \eqref{eq:l1-synthesis}, which complicates the derivation of the estimates.
\end{remark*}

\begin{remark*}
When $J$ is infinite, the above minimization problem cannot be implemented numerically. When $D$ and $U$ are unitary operators, it was shown in \cite{AH} that this issue may be solved by looking at a corresponding finite-dimensional optimization problem. We expect that the same is true also in our context, and leave this investigation to future work.
\end{remark*}

Given the generality of our setting, we need to consider the dual frames of $\{\psi_l\}_l$ and $\{\varphi_j\}_j$. By classical frame theory (see \cite[Lemma 5.1.5]{2016-christensen}), the \emph{frame operators} $U^*U$ and $D^*D$ are invertible, and we can consider the \emph{dual frames}
\[
\tilde\psi_l := (U^*U)^{-1}\psi_l\qquad\text{and}\qquad 
\tilde\varphi_j := (D^*D)^{-1}\varphi_j,
\]
which have frame constants $B_U^{-1},A_U^{-1}$ and $B_D^{-1},A_D^{-1}$, respectively. Equivalently, we may write $\tilde\psi_l = U^{-1}e_l$ and $\tilde\varphi_j = D^{-1}e_j$, where $U^{-1}$ and $D^{-1}$ are the Moore--Penrose pseudoinverses of $U$ and $D$, respectively, defined as follows:
\begin{equation*}
U^{-1}:=(U^*U)^{-1}U^* \quad \text{and} \quad D^{-1}:=(D^*D)^{-1}D^*.
\end{equation*}
Note that they are left inverses of $U$ and $D$, respectively. 
Therefore, $(U^{-1})^*$ and $(D^{-1})^*$ are the analysis operators of the dual frames, and so arguing as in  \eqref{eq:normUD} we obtain
\begin{equation}\label{eq:normU-D-}
\norm{U^{-1}}=\norm{U^{-*}}\le A_U^{-1/2},\qquad 
\norm{D^{-1}}=\norm{D^{-*}}\le A_D^{-1/2}.
\end{equation}
With an abuse of notation, we have denoted $(U^{-1})^*$ and $(D^{-1})^*$ by $U^{-*}$ and $D^{-*}$, respectively. It can be immediately  checked that they are right inverses of $U^*$ and $D^*$, i.e. $(U^{*})^{-1}=U^{-*}$ and $(D^*)^{-1}=D^{-*}$. For later use, set $\k_1 := \max(B_U,A_U^{-1})$ and $\k_2 := \max(B_D,A_D^{-1})$, so that by \eqref{eq:normUD} and \eqref{eq:normU-D-} we obtain
\begin{subequations}\label{eq:boundk}
\begin{align}\label{eq:boundkU}
&\norm{U}=\norm{U^*}\le \sqrt{\k_1}, 
&&\norm{U^{-1}}=\norm{U^{-*}}\le \sqrt{\k_1}, \\
\label{eq:boundkD}
&\norm{D}=\norm{D^*}\le \sqrt{\k_2},&&\norm{D^{-1}}=\norm{D^{-*}}\le \sqrt{\k_2}.
\end{align}
\end{subequations}

The frames $\{\tilde\psi_l\}_l$ and $\{\tilde\varphi_j\}_j$ are the \emph{canonical} dual frames, but in general many other choices are possible. These are in correspondence with all possible bounded left inverses of $U$ and $D$, and it is possible to  characterize all dual frames  \cite[Section 6.3]{2016-christensen}.

We need to measure the coherence between the sensing system $\{\psi_l\}_l$ and the representation system $\{\varphi_j\}_j$ or, equivalently, between the measurement operator $U$ and the representation operator $D$.

\begin{definition}\label{def:mu}
The \emph{mutual coherence} of $U$ and $D$ is defined by
\[
\begin{split}
\mu&:=\sup_{{j\in J,\,l\in L}}\max\{|\langle \varphi_j,\psi_l\rangle_\H|,
|\langle \tilde\varphi_j,\psi_l\rangle_\H|,
|\langle \varphi_j,\tilde\psi_l\rangle_\H|,
|\langle \tilde\varphi_j,\tilde\psi_l\rangle_\H|  \}\\
&{\small
\text{$
\;=\sup_{{j\in J,\,l\in L}}\max\{|\langle D^*e_j,U^{*}e_l\rangle_\H|,
|\langle D^{-1}e_j,U^{*}e_l\rangle_\H|,
|\langle D^*e_j,U^{-1}e_l\rangle_\H|,
|\langle D^{-1}e_j,U^{-1}e_l\rangle_\H|  \}.
$}}
\end{split}
\]
\end{definition}

Let us now discuss a particular case.

\begin{example}\label{ex:isometry}
The above construction simplifies considerably if $\{\psi_l\}_l$ and $\{\varphi_j\}_j$ are \emph{Parseval} frames, namely if $A_U=B_U=A_D=B_D=1$, as studied in \cite{Poon2015}. In this case the associated analysis operators $U$ and $D$ are isometries, their left inverses simplify to $U^{-1}=U^*$, $D^{-1}=D^*$ and all the operator norms in \eqref{eq:boundk} are simply bounded by $1$. The dual frames and the corresponding frames coincide, and the coherence reduces to
\begin{equation}\label{eq:mu-bases}
\mu= \sup_{{j\in J,\,l\in L}}|\langle \varphi_j,\psi_l\rangle_\H|=\sup_{{j\in J,\,l\in L}} |\langle D^*e_j,U^{*}e_l\rangle_\H| ,
\end{equation}
which simply involves scalar products between the elements of the two frames.

As an even more particular case, one may consider orthonormal bases $\{\psi_l\}_l$ and $\{\varphi_j\}_j$ (namely, $U$ and $D$ are unitary operators), which represents the usual assumption in the classical compressed sensing framework, and in its extension to infinite dimension \cite{AH}.
\end{example}

The partial measurements  $(Ug_0)_l=\langle g_0,\psi_l\rangle_\H$ are indexed by $l\in\Omega$, where $\Omega$ will be chosen uniformly at random in $\{1,\dots,N\}$. In the infinite dimensional case, the upper bound $N$ has to be chosen big enough, depending on the sparsity $s$ of $g_0$. This is quantified by the \emph{balancing property}, introduced in \cite{AH} and extended here to general frames.

\begin{definition}[Balancing property]\label{bp}
Let $s,M \in J$ be such that $2\le  s\le M$. We say that $N \in L$ satisfies the balancing property with respect to $U$, $D$,  $M$ and $s$ if for all $\Delta \subseteq \{1,\ldots,M\}$ with $|\Delta|=s$ we have
\begin{align}
\label{bp1}
&\Vert P_{\W} U^{*} P_N^\perp \U  P_{\W} \Vert_{\H \to \H}  \leq \frac{1}{8 \sqrt{\sqrt{\k_2}\log (s\k_1^2\k_2)} },\\
&\|P_\Delta^\perp D^{-*} P_{\W}^\perp U^* P_N \U P_{\W}\|_{\H \to \ell^{\infty}{(J)}} \leq \frac{1}{14\sqrt{s\k_2}}, \label{bp2} 
\end{align}
where $\W:=R(D^{\ast} P_{\Delta})+R(D^{-1} P_{\Delta})=\{D^*P_\Delta x+D^{-1}P_\Delta y:x,y\in\ell^2{(J)}\}$.
\end{definition}

\begin{remark}\label{rem:balancing}
If $L=\{1,2,\dots,|L|\}$ is finite, it is enough to choose $N=|L|$, since all the norms on the left hand side vanish. If $L=\N$, the existence of a suitable $N$ satisfying the above conditions simply follows by the fact that $P_N\to I$ and $P_N^\perp\to0$ strongly (see \cite[Proposition 5.2]{AH} for the details of the argument). 
\end{remark}

\begin{remark}\label{rem:balancing2}
In many particular cases of practical interest, it is possible to find explicit bounds on $N$. For instance, when $D$ is a discrete wavelet transform (satisfying certain assumptions) and $U$ is the discrete Fourier transform on $\H=L^2([0,1])$ one can prove that $N=O(M\log M)$ \cite{poon2014,2017-adcock-hansen-poon-roman}.
\end{remark}

For $s,M\in {J}$, $s\le M$, we use the notation 
\[
\sigma_{s,M}(x_0):=\inf\{\norm{x-x_0}_{\ell^1 {(J)}}: \supp(x)\subseteq\{1,\dots,M\},\, |\supp(x)|\le s\},
\]
which measures the compressibility of the signal $x_0\in\ell^1 {(J)}$ by means of $s$-sparse signals $x$. As in \cite{2015-krahmer-needell-ward},  we introduce the localization factor
 \[
\eta_{s,M} := \max\{\eta_\Delta:\Delta\subseteq\{1,\dots,M\},\;3\le|\Delta|\le s\},
\]
where
\begin{equation}\label{eq:eta}
\eta_\Delta := \max_{i=0,1} \,\sup\,(\{\frac{\|D_i D_i^* x\|_1}{\sqrt{|\Delta|}}: \supp\, x \subseteq\Delta,\;\|D^\ast_i x\|_\H = 1\}\cup\{1\}),
\end{equation}
$D_0:=D$ and $D_1:=D^{-*}$.
Note that $\eta_{s,M} =1$  when $\{\varphi_j\}_j$ is an orthonormal basis or, equivalently, when $D$ is a unitary operator. 

Following \cite{Poon2015}, for $\Delta\subseteq\{1,\dots,M\}$ we denote
\begin{equation}\label{eq:defB}
\begin{aligned}
&B_\Delta := \max \left\{\|D^{-*} P_{\W}^\perp D^{*}\|_{\ell^\infty \to \ell^{\infty}},1\right\},\\
&B_{s,M} := \max\{B_\Delta:\Delta\subseteq\{1,\dots,M\},\;3\le  |\Delta|\le s\},
\end{aligned}
\end{equation}
where we have used the notation
\[
\| T  \|_{\ell^\infty \to \ell^{\infty}}:=\sup_{x\in\ell^2(J)\setminus\{0\}} \frac{\|Tx\|_{\ell^\infty {(J)}}}{\|x\|_{\ell^\infty {(J)}}},
\]
for an operator $T\colon \ell^2({J})\to\ell^2({J})$.
\begin{remark}\label{rem:B}
It is worth observing that when $D$ is a unitary operator we simply have
\[
B_{s,M}=1.
\]
Indeed, in view of the identity
\[
DP_\W^\perp D^*x=DP_\W^\perp D^*(P_\Delta x+P_\Delta^\perp x)=DP_\W^\perp D^*P_\Delta^\perp x=DD^*P_\Delta^\perp x=P_\Delta^\perp x,
\] 
we obtain $\|D P_{\W}^\perp D^{*}x\|_\infty=\norm{P_\Delta^\perp x}_\infty\le \norm{x}_\infty$ for every $x\in\ell^2({J})$, so that $B_\Delta= 1$ for every $\Delta$.
\end{remark}
In the other extreme case, it may happen that $B_{s,M}=+\infty$, even in finite dimension with a Parseval frame, as the following example shows. 
\begin{example}
Consider $\H=\R$ with the  Parseval frame
\[
\varphi_1=\varphi_2=\varphi_3 = 0, \qquad\varphi_{j+3}= f_j,\;j\ge 1,
\]
where $f\colon\N\to (0,+\infty)$ is a sequence such that $\sum_j f_j^2=1$ and $\sum_j f_j=+\infty$ {(here $J=\N$)}. For $\Delta=\{1,2,3\}$ we have $\W=\{0\}$, so that 
$
B_\Delta \ge \|D P_{\W}^\perp D^{*}\|_{\ell^\infty \to \ell^{\infty}}=
\|D  D^{*}\|_{\ell^\infty \to \ell^{\infty}}
$. Thus, setting $x_n=e_1+\dots+e_{n+3}\in\ell^2(\N)$, since $D^*x_n = \sum_{j=1}^n f_j$ we have
\[
B_\Delta \ge |(DD^*x_n)_4|=|\langle D^*x_n,\varphi_4\rangle_\R|=f_1\sum_{j=1}^n f_j\underset{n\to \infty}{\longrightarrow} +\infty,
\]
whence $B_{s,M}\ge B_\Delta\ge +\infty$ for any $s$ and $M$.
\end{example}
These observations show that the quantities $\eta_{s,M}$ and  $B_{s,M}$  can be seen as measures of the redundancy of the frame $\{\varphi_j\}_j$ (see \cite{2015-krahmer-needell-ward,Poon2015} for more details).

For $\alpha\in (0,1]$, let $\tilde M(\alpha)$ be the smallest integer such that $\tilde M(\alpha)\ge M$ and
\begin{equation}\label{eq:tildeM}
\sqrt{\k_1}\| P_N U D^{-1} e_j\|_2 + \k_1\| P_{\widetilde\W} D^{-1} e_j\|_{\H}< \alpha,\qquad j\in J,\;j> \tilde M(\alpha),
\end{equation}
where $\widetilde\W := R(D^{\ast} P_{M})+R(D^{-1} P_{M})$.
\begin{remark}\label{rem:Mtilde}
If $J=\{1,2,\dots,M\}$  is finite, we simply have $\tilde M(\alpha)=M$ for every $\alpha$. If $J=\N$, $\tilde M(\alpha)$ is well-defined since $D^{-1} e_j$ tends to zero weakly and $P_N$ and $P_{\widetilde\W}$ are compact operators. 
\end{remark}
\begin{remark}\label{rem:M}
In the case when $D$ is associated with an orthonormal basis, the condition $\tilde M(\alpha)\ge M$ is implicit, since  $\k_1\| P_{{\widetilde\W}} D^\ast e_j\|_\H = \k_1\| D^\ast e_j\|_\H=\k_1\ge 1\ge\alpha $ for $j=1,\dots,M$ by definition of $\widetilde \W$. Furthermore, condition \eqref{eq:tildeM} reduces to
\begin{equation*}
\sqrt{\k_1}\| P_N U D^\ast e_j\|_2< \alpha,\qquad j\in J,\;j> \tilde M(\alpha).
\end{equation*}
{As a consequence, note that if $ \sup_{l\le N} |\langle\psi_l,\varphi_j\rangle_\H| \le C/\sqrt j$ for every $j\in J {, j>M}$ (which is the case in several concrete applications, see {$\S$\ref{sub:wavelets}}) one has 
\[
\sqrt{\k_1}\| P_N U D^\ast e_j\|_2 \le \sqrt{N\k_1}\| P_N U D^\ast e_j\|_\infty
\le  \sqrt{N\k_1} \sup_{l\le N} |\langle\psi_l,\varphi_j\rangle_\H|
\le  C \sqrt{N\k_1/j},
\]
and so
\[
\tilde M(\alpha) \le C^2\frac{\k_1 N}{\alpha^2},
\]
provided that $C^2\frac{\k_1 N}{\alpha^2}\ge M$.
In the case where $U$ is the Fourier transform and $D$ a Wavelet transform in dimension one, a more precise estimate has been derived in  {\cite{2017-adcock-hansen-poon-roman}}, namely $\tilde M(\alpha) = O(M/\alpha)$.
}
\end{remark}

\section{Main results}\label{sec:thm}

\subsection{Finite and infinite dimensional recovery}\label{sec:main}

We now state the main result of this work, a recovery guarantee of nonuniform type (valid for a fixed unknown signal). Recall that $\H$ is any separable Hilbert space: we deal with the finite and infinite dimensional case simultaneously.
\begin{theorem}\label{thm:main}
Assume that Hypothesis~\ref{hp1} holds true, and let $U$ and $D$ denote the corresponding analysis operators, satisfying the bounds given in \eqref{eq:boundk}. Let $M,s\in J$ and $\omega \geq 1$ be such that $3\le s\le M$. Let $N \in L$ satisfy the balancing property with respect to $U$, $D$, $M$ and $s$, and let $\Omega \subseteq \{1, \ldots, N\}$ be chosen uniformly at random with $|\Omega|=m$. Assume that
\[
m \geq C \k_1\k_2 B_{s,M}^2 \eta_{s,M}^2 \omega^2 \mu^2   N s \log\left(\k_1\k_2 \tilde M\bigl(\tfrac{C'm}{N\omega\sqrt{s\k_2}}\bigr)\right),
\]
where $C,C'>0$ are  universal constants.

Let $g_0 \in \H$ and $\eta \in\ell^2({L})$ be such that $\norm{\eta}_2\le\epsilon$ for some $\epsilon\ge 0$. Let $\zeta = P_{\Omega} U g_0 + \eta$ be the known noisy measurement.
Let $g \in \H$ be a minimizer of the minimization problem \eqref{eq:l1}.
Then, with probability exceeding $1-e^{-\omega}$, we have
\[
\Vert g - g_0 \Vert_\H \leq 20\k_1\sqrt{\k_2} \,\sigma_{s,M}(Dg_0) +  C'' \k_2\sqrt{\k_1^3\omega s\frac{N}{m}}   \epsilon,
\]
where $C''$ is a universal constant. 
\end{theorem}

\begin{remark*}
The generality of our construction allows to treat the finite dimensional and the infinite dimensional cases simultaneously. However, in finite dimension the above estimate for $m$ has a simpler form, which is worth pointing out. Suppose that $L=\{1,\dots,N\}$ and $J=\{1,\dots,M\}$. By Remarks~\ref{rem:balancing} and~\ref{rem:Mtilde}, we have that $N$ satisfies the balancing property with respect to $U$, $D$, $M$ and $s$ and that $\tilde M\bigl(\tfrac{C'm}{N\omega\sqrt{s\k_2}}\bigr)=M$. Thus, the lower bound for the number of measurements $m$ becomes
\[
m \geq C   \k_1\k_2   B_{s,M}^2 \eta_{s,M}^2 \omega^2 \mu^2 N  s \log\left(\k_1\k_2 M\right),
\]
which, when  $\{\varphi_j\}_j$ is an orthonormal basis of $\H=\C^M$, by  \eqref{eq:boundkD}, Example~\ref{ex:isometry} and Remark~\ref{rem:B} simply reduces to
\[
m \geq C   \k_1 \omega^2   \mu^2 Ns \log (\k_1 M).
\]
\end{remark*}

Theorem~\ref{thm:main} directly generalizes Theorems 6.1, 6.3 and 6.4 of \cite{AH} to the case of anisotropic measurements. It also extends the results of \cite{2017-adcock-hansen-poon-roman,Poon2015} to the case of general frames $D$ and $U$.  Our result can also be seen as an infinite dimensional generalization of the finite dimentional result in \cite{KG} for anisotropic random measurements. 

\subsection{Asymptotic incoherence and virtual frames}\label{sub:asymptotic}

The above result shows that with random sampling one needs a number of measurements proportional to the sparsity of the signal (up to logarithmic factors), provided that the coherence is small enough, namely, $\mu=O(1/\sqrt{N})$. While this happens in finite dimension ($\H=\C^N$) in some particular situations, e.g.\ with signals that are sparse with respect to the Dirac basis $\varphi_j=e_j$ and with Fourier measurements, in many cases of practical interest the above result becomes almost meaningless since the coherence $\mu$ is of order one. For instance, this happens when $U$ is the discrete Fourier transform and $D$ the discrete wavelet transform. As it was shown in \cite{2017-adcock-hansen-poon-roman}, this is always the case in infinite dimension.

Since the early stages of compressed sensing, it was realized that this issue may be solved by using variable density random sampling \cite{2006-tsaig-donoho,2011-puy-v-wiaux,2014-krahmer-ward,2015-krahmer-needell-ward,2016-bigot-boyer-weiss,2017-adcock-hansen-poon-roman,Poon2015}. For instance, in the Fourier-Wavelet case, one needs to sample lower frequencies with higher probability than the higher frequencies. We now give a result that deals with this situation; in particular, it takes into account a priori estimates on the coherence and nonuniform sampling. As it is clear from the proof, it follows as a simple corollary of Theorem~\ref{thm:main}, thanks to the flexibility of its assumptions. More precisely, Theorem~\ref{thm:main} is applied to a \textit{virtual frame} $\{\hat \psi_{\hat l}\}_{\hat l}$ obtained from $\{\psi_l\}_l$ by artificially repeating its elements. In the Appendix we justify how the uniform sampling on this virtual frame is equivalent to a nonuniform sampling on the original frame. More complicated transformations, also involving $\{\varphi_j\}_j$, may be considered (taking into account, for instance, \emph{asymptotic sparsity} \cite{2017-adcock-hansen-poon-roman}): we leave these investigations to future work, and we limit ourselves to an example to show the potential of this framework.

\begin{cor}\label{cor:main}
Assume that Hypothesis~\ref{hp1} holds true, and let $U$ and $D$ denote the corresponding analysis operators, satisfying the bounds given in \eqref{eq:boundk}. Let $M,s\in {J}$ and $\omega \geq 1$ be such that  $3\le s\le M$. Let $N \in L$ satisfy the balancing property with respect to $U$, $D$, $M$ and $s$. Let $w \in\R_+^N$ be such that $\norm{w}_{\C^N}\ge 1$ and
\begin{equation}\label{def:async}
\sup_{j \in J} \max\{|\langle \varphi_j,\psi_l\rangle_\H|,
|\langle \tilde\varphi_j,\psi_l\rangle_\H|,|\langle \varphi_j,\tilde\psi_l\rangle_\H|,
|\langle \tilde\varphi_j,\tilde\psi_l\rangle_\H| \} \leq w_l, \qquad l=1,\dots,N.
\end{equation}
Assume that
\begin{equation}\label{eq:mcor}
m \geq C \k_1\k_2   B_{s,M}^2 \eta_{s,M}^2 \omega^2 \|w\|_{\C^N}^2  s\log\Bigl(\k_1\k_2 \tilde M\bigl(\tfrac{C'}{Nm \omega\|w\|_{\C^N}^2\sqrt{s\k_2}}\bigr)\Bigr),
\end{equation}
where $C,C'>0$ are universal constants.
Sample $m$ indices $l_1, \ldots, l_m$ indipendently from $\{1,\ldots,N\}$  according to the probability distribution
\[
 \nu_l = C_N \lceil N w_l^2\rceil, \qquad l=1,\dots,N,
 \]
where $C_N = \left(\sum_{l=1}^N \lceil N w_l^2\rceil \right)^{-1}$, and set $\Omega = \{l_1,\ldots,l_m\}$ (with possible repetitions to be kept).

Take $g_0 \in \H$ and $\eta \in \C^m$ such that {$\| \eta \|_w \leq \varepsilon$}, where  $\|\eta\|^2_w : = \sum_{i=1}^m \frac{|\eta_i|^2}{\lceil N w_{l_i}^2 \rceil}$. Set $\zeta = P_{\Omega}U g_0+\eta$, i.e.\ $\zeta_i = (Ug_0)_{l_i}+\eta_i$. Let $g \in \H$ be a minimizer of
\begin{equation*}\label{eq:l1w}
\inf_{ \substack{
 g \in \H \\  D  g \in \ell^1({J})
}} \Vert D  g \Vert_{1} \quad \text{subject to } \|P_\Omega U  g - \zeta\|_w \le\epsilon.
\end{equation*}
Then, with probability exceeding $1-e^{-\omega}$, we have
\[
\Vert g - g_0 \Vert_\H \leq 20\k_1\sqrt{\k_2} \,\sigma_{s,M}(Dg_0) +  C'' \k_2 \|w\|_{\C^N}\sqrt{\k_1^3\omega s\frac{N}{m}}   \epsilon,
\]
where $C''$ is a universal constant.
\end{cor}

\begin{remark*}
This result can be seen as a generalization of \cite[Corollary 2.9]{2015-krahmer-needell-ward} to infinite dimension and to the frame case, since $\nu_l = \lceil N w_l^2\rceil / \sum_{l=1}^N \lceil N w_l^2\rceil \approx \frac{w_l^2}{\|w\|_{\C^N}^2}$.
\end{remark*}

\begin{remark*}
Let us show that the  weighted norm $\norm{\;}_w$ is the natural norm in this context. If $l\in\{1,\dots,N\}$ is sampled according to $\nu$, we have
$
\E\bigl(\frac{|v_l|^2}{r_l}\bigr) = \sum_{i=1}^N \frac{|v_i|^2}{r_i} \nu_i= C_N\norm{v}_{\C^N}^2
$
for every $v\in\C^N$. Hence
\[
\mathbb{E}(\norm{P_\Omega v}_w^2) 
= \E\left(\sum_{i=1}^m \frac{|v_{l_i}|^2}{r_{l_i}}\right) = \sum_{i=1}^m\E\left( \frac{|v_{l_i}|^2}{r_{l_i}}\right) = m C_N \norm{v}_{\C^N}^2,\qquad v\in\C^N.
\]
As a consequence, up to a constant, the expected value of the weighted norm coincides with the usual Euclidean norm of $\C^N$.
\end{remark*}

\begin{remark*}The bound \eqref{def:async} may be completed by the corresponding decays with respect to the frame $\{\varphi_j\}_j$, namely in the variable $j$. In this case, $\{ \psi_l\}_{l \in L}$ and $\{ \varphi_j\}_{j \in J}$ are \emph{asymptotically incoherent} \cite{2017-adcock-hansen-poon-roman}. Under this more restrictive assumption, when $D$ is a unitary operator an explicit bound on the factor $\tilde M$ may be derived using Remark~\ref{rem:M} (see $\S$\ref{sub:wavelets} for an example).
\end{remark*}

\begin{proof}
For $l\in L$, $l>N$, set
\begin{equation}\label{eq:wl}
w_l=\sup_{j\in J} \max\{|\langle \varphi_j,\psi_l\rangle_\H|,
|\langle \tilde\varphi_j,\psi_l\rangle_\H|,|\langle \varphi_j,\tilde\psi_l\rangle_\H|,
|\langle \tilde\varphi_j,\tilde\psi_l\rangle_\H|  \}  .
\end{equation}
Let $r_l = \upsilon\lceil  N w_l^2 \rceil$ for $l \in L$, where $\upsilon=2m^2$. We want to apply Theorem~\ref{thm:main} to $\{\hat \psi_{\hat l}\}_{\hat l}$ and $\{\varphi_j\}_j$, where the virtual frame $\{\hat \psi_{\hat l}\}_{\hat l}$, with associated operator $\hat U$, is given as follows. For $l \in L$ normalize $\psi_l$  by $ \sqrt{r_l}$ and  repeat it $r_l$  times, namely
\begin{equation}\label{eq:virtual}
\left\{\hat \psi_{\hat l} \right\}_{{\hat l\in\hat L}}=\biggl\{ \underbrace{\frac{\psi_1}{\sqrt{r_1}}, \ldots, \frac{\psi_1}{\sqrt{r_1}}}_{r_1 \text{ times}}, \ldots, \underbrace{\frac{\psi_l}{\sqrt{r_l}}, \ldots, \frac{\psi_l}{\sqrt{r_l}}}_{r_l \text{ times}},\ldots\biggr\}.
\end{equation}
The new index set $\hat L$ coincides with $\N$ if $L=\N$ and is finite if $L$ is finite (more precisely,  we have $|\hat L| = \sum_{l\in L}r_l$).

Note that $\{\hat \psi_{\hat l}\}_{\hat l}$ has the same frame bounds of $\{ \psi_l\}_l$, i.e. {$\hat \k_1 = \k_1$}, by construction, since
\begin{equation}\label{eq:pars}
\sum_{i=1}^{r_l}\left|\left\langle f,\frac{\psi_l}{\sqrt{ r_l}}\right\rangle\right|^2 = \sum_{i=1}^{r_l}\frac{1}{r_l}|\langle f,\psi_l\rangle|^2 = |\langle f,\psi_l\rangle|^2.
\end{equation}

We now want to prove that $\hat N = \sum_{l=1}^N r_l$ satisfies the balancing property with respect to $\hat U$, $D$, $M$ and $s$. We first notice that $\hat U^\ast \hat U = U^\ast U$, since 
\[
\hat U^\ast \hat U f = \sum_{{\hat l\in\hat L}}\langle f,\hat \psi_{\hat l}\rangle\hat \psi_{\hat l} = \sum_{{l\in L}} \sum_{i=1}^{r_l}\left\langle f, \frac{\psi_l}{\sqrt{r_l}}\right\rangle \frac{\psi_l}{\sqrt{r_l}} 
= \sum_{{l\in L}} \langle f,\psi_l\rangle\psi_l = U^\ast Uf.
\]
In passing, we remark that this identity tells us that the dual frame $\{\tilde{\hat\psi}_{\hat l} \}_{\hat l}$ of the virtual frame $\{ \hat \psi_{\hat l}\}_{\hat l}$ coincides with the virtual frame of the dual frame $\{\hat{\tilde{\psi}}_{\hat l} \}_{\hat l}$, which is constructed as in \eqref{eq:virtual}. Arguing in the same way, and terminating the above sums to $\hat N$ and $N$, respectively,  we readily derive
\[
\hat U^\ast P_{\hat N} \hat U f 
= \sum_{{\hat l}=1}^{\hat N}\langle f,\hat \psi_{\hat l}\rangle\hat \psi_{\hat l}
= \sum_{l=1}^N \sum_{i=1}^{r_l}\left\langle f, \frac{\psi_l}{\sqrt{r_l}}\right\rangle \frac{\psi_l}{\sqrt{r_l}} 
= \sum_{l=1}^N \langle f,\psi_l\rangle\psi_l = U^\ast P_N Uf.
\]
This immediately  yields property \eqref{bp2}, since {$\hat U^{-*} = \hat U (\hat U^\ast \hat U)^{-1}$ and $\hat U^{-1} =(\hat U^\ast \hat U)^{-1} \hat U^\ast$}. Similarly, \eqref{bp1}  follows from the identities
\[
\hat U^\ast P^\perp_{\hat N} \hat U = \hat U^\ast (I - P_{\hat N}) \hat U = \hat U^\ast \hat U - \hat U^\ast P_{\hat N} \hat U = U^\ast U - U^\ast P_N U  = U^\ast P_N^\perp U.
\]

We have the following straightforward upper bound for $\hat N$:
\[
\hat N  = \sum_{l=1}^N\u\lceil  N w_l^2 \rceil \leq \u\sum_{l=1}^N(N w_l^2+1) = \u N(\|w\|_{\C^N}^2+1)\le 2\u N\|w\|_{\C^N}^2.
\]

The factor $\tilde M$ associated with $\hat U$ and $D$, which we denote by  $\hat{\tilde M}$, verifies $\hat{\tilde M}({\alpha})= {\tilde M}({\alpha})$. Indeed, from the definition of $\tilde M$, we only need to check that $\| P_{\hat N} \hat U f\|_2^2=\|P_N Uf\|_2^2$, which follows by \eqref{eq:pars}:
\begin{align*}
\| P_{\hat N} \hat U f\|_2^2 &= \sum_{{\hat l}=1}^{\hat N} |\langle\hat U f,e_{\hat l}\rangle|^2 = \sum_{{\hat l}=1}^{\hat N} |\langle f,\hat \psi_{\hat l}\rangle|^2=\sum_{l=1}^N|\langle f,\psi_l\rangle|^2 = \|P_N Uf\|_2^2.
\end{align*}
The factors $B_{s,M}$ and $\eta_{s,M}$ do not change since they do not depend on $U$ but only on $D$, which is left unchanged. 

Let us calculate the new coherence 
\[
\hat \mu = \sup_{{\hat l},j} \max\{|\langle \varphi_j,\hat \psi_{\hat l}\rangle_\H|,
|\langle \tilde\varphi_j,\hat\psi_{\hat l}\rangle_\H|,|\langle \varphi_j,\tilde{\hat{\psi_l}}\rangle_\H|,
|\langle \tilde\varphi_j,\tilde{\hat{\psi}}_{\hat l}\rangle_\H| \}.
\] 
For $\hat l \in\hat L$ there exists $l\in L$ with $\hat \psi_{\hat l} = \psi_{l}/\sqrt{r_{l}}$ and $\tilde{\hat{\psi}}_{\hat l} = \tilde \psi_{l}/\sqrt{r_{l}}$, so that by \eqref{def:async} and \eqref{eq:wl} we obtain
\begin{multline*}
\max\{|\langle \varphi_j,\hat \psi_{\hat l}\rangle_\H|,
|\langle \tilde\varphi_j,\hat\psi_{\hat l}\rangle_\H|,|\langle \varphi_j,\tilde{\hat{\psi_l}}\rangle_\H|,
|\langle \tilde\varphi_j,\tilde{\hat{\psi}}_{\hat l}\rangle_\H| \} \\
=\max\{|\langle \varphi_j, \frac{\psi_{l}}{\sqrt{r_{l}}}\rangle_\H|,
|\langle \tilde\varphi_j,\frac{\psi_{l}}{\sqrt{r_{l}}}\rangle_\H|,
|\langle \varphi_j,\frac{\tilde \psi_{l}}{\sqrt{r_{l}}}\rangle_\H|,
|\langle \tilde\varphi_j,\frac{\tilde \psi_{l}}{\sqrt{r_{l}}}\rangle_\H| \}  \leq \frac{w_{l}}{\sqrt{r_{l}}} \leq \frac{1}{\sqrt{\u N}},
\end{multline*}
since $r_l= \u\lceil N w_l^2 \rceil$. Therefore $\hat \mu \leq \frac{1}{\sqrt{\u N}}$.

Now, the factor $\hat \mu^2 \hat N$ in the estimate for $m$ given in Theorem~\ref{thm:main} applied to $\hat U$ and $D$ becomes $\hat \mu^2 \hat N \leq 2\|w\|_{\C^N}^2$, so that the estimate on $m$ in Theorem~\ref{thm:main} transforms into \eqref{eq:mcor}. 

Further, setting $\hat \zeta = P_{\hat\Omega}\hat U g_0 + \left({\eta_i}/{\sqrt{r_{l_i}}}\right)_i$, we have
\[
\norm{P_{\hat\Omega}\hat U g-\hat\zeta}_2
=\norm{\left(\left( (g,\psi_{l_i})-{\zeta_i}\right)/{\sqrt{r_{l_i}}}\right)_i}_2
=\frac{\norm{P_{\Omega} U g-\zeta}_w}{\sqrt{\u}}\le\frac{\epsilon}{\sqrt{\u}}=:\hat\epsilon.
\]
As a result, the factor $\sqrt{\hat N}\,\hat \epsilon$ in the estimate for $\Vert g - g_0 \Vert_\H$ given in Theorem~\ref{thm:main} is bounded by  $\|w\|_{\C^N}\sqrt{ 2N}\, \epsilon$.

Finally, it is well known that the sampling without replacement of $m$ elements $\Omega'$ uniformly at random from \[
\biggl\{ \underbrace{1, \ldots, 1}_{r_1 \text{ times}}, \ldots, \underbrace{N, \ldots, N}_{r_N \text{ times}}\biggr\}
\]
 given in Theorem~\ref{thm:main} may be substituted by the variable density sampling scheme with replacement given in the statement, provided that $\u $ is large enough. More precisely, letting  $n_l$ (respectively, $n'_l$) denote the number of $l$'s in $\Omega$ (respectively, $\Omega'$) for $l=1,\dots,N$, by arguing as in \cite[Section II.C]{CRT}, by Lemma~\ref{lem:ernesto} (in the Appendix below) we readily derive
\[
\begin{split}
\P(F(\Omega'))&=\sum_{ \substack{
 k_1,\dots,k_N=0 \\ k_1+\dots+k_N=m }}^m \P(F(\Omega')|n'_1=k_1,\dots,n'_N=k_N)\,\P(n'_1=k_1,\dots,n'_N=k_N) \\
 &=\sum_{  k_1+\dots+k_N=m } \P(F(\Omega)|n_1=k_1,\dots,n_N=k_N)\,\P(n'_1=k_1,\dots,n'_N=k_N) \\
 &\ge \frac12 \sum_{ k_1+\dots+k_N=m } \P(F(\Omega)|n_1=k_1,\dots,n_N=k_N)\,\P(n_1=k_1,\dots,n_N=k_N)\\
 &= \frac12 \P(F(\Omega)),
\end{split}
\]
where $F(\Omega)=\operatorname{Failure}(\Omega)$ is the event where the conclusion of Theorem~\ref{thm:main} fails for the sampling $\Omega$.

\end{proof}

\subsection{Recovery of wavelet coefficients from Fourier samples}\label{sub:wavelets}

 For $d\in\N$, let $\H=L^2([0,1]^d)$  be the signal space. Let $\{e_k\}_{k\in\Z^d}$ be the Fourier basis of $\H$, namely
\[
e_k(x) =e^{2\pi ik\cdot x},\qquad x\in [0,1]^d.
\]
Consider a nondecreasing ordering of $\Z^d$, namely a bijective map $\rho\colon\N\to \Z^d$, $l\mapsto k_l$, such that
\[
l_1,l_2\in\N,\; l_1\le l_2  \quad\implies\quad \norm{\rho(l_1)}\le\norm{\rho(l_2)},
\]
where $\norm{\;}$ is any norm of $\R^d$. Set $\psi_l = e_{k_l}$ for $l\in\N$. Let $\{\varphi_j\}_{j\in\N}$ be a separable wavelet basis of $\H$ (ordered according to the wavelet scales). Note that both systems $\{\psi_l\}_{l\in\N}$ and $\{\varphi_j\}_{j\in\N}$ are orthonormal bases, so that $\tilde\psi_l=\psi_l$ and $\tilde\varphi_j=\varphi_j$. Under certain decay conditions on the scaling function (which may be relaxed {to a condition satisfied by all Daubechies wavelets} if one considers a different ordering of the frequencies $k\in\Z^d$), it was shown in \cite{2016arXiv161007497J} that
\[
\sup_{j\in\N} |\langle\psi_l,\varphi_j\rangle_\H| \le \frac{C_1}{\sqrt{l}},\quad 
\sup_{l\in\N} |\langle\psi_l,\varphi_j\rangle_\H| \le \frac{C_1}{\sqrt{j}},\qquad l,j \in\N
\]
for some $C_1>0$. In other words, the wavelet basis and the Fourier basis are asymptotically incoherent. Thanks to the first of these inequalities, we have that assumption \eqref{def:async} of the corollary is satisfied with $w_l = \frac{C_1}{\sqrt{l}}$, so that $\norm{w}_{\C^N}^2\le C_1^2(\log N+1)$. Further, by the second of these inequalities and Remark~\ref{rem:M} we have
$
\tilde M(\alpha) \le C_1^2\frac{ N}{\alpha^2}.
$ 
As a consequence, estimate \eqref{eq:mcor} {of Corollary~\ref{cor:main}} for the number of measurements $m$ becomes
\begin{equation}\label{eq:m}
m \geq C   \, \omega^2 s \log^2 N   
\end{equation}
for some constant $C>0$ depending only on $C_1$. The number of measurements required for the success of the recovery using $\ell^1$ minimization is directly proportional to the sparsity of the signals, up to log factors, and so, up to log factors,  estimate \eqref{eq:m} is optimal. It is worth observing that one log factor may be removed by using multilevel sampling, asymptotic sparsity and finer properties of wavelets \cite[Theorem~6.2]{2017-adcock-hansen-poon-roman}.

\begin{figure}
\begin{picture}(300,200)
\put(-37,0){
\includegraphics[width=\linewidth]{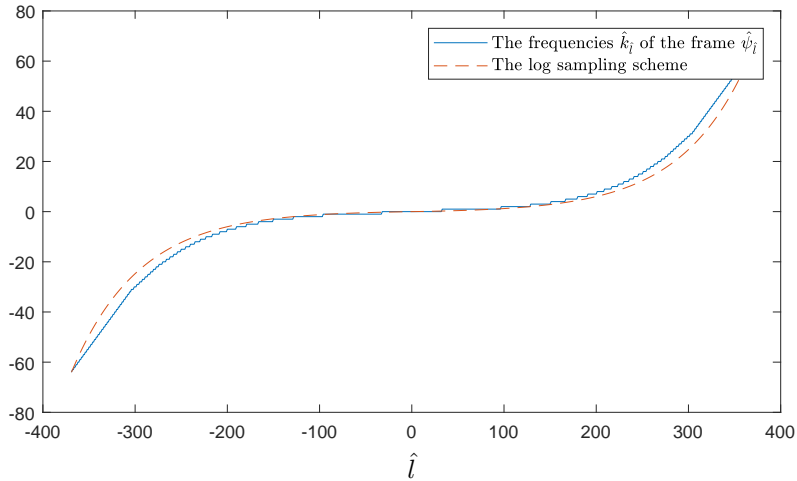}}
\put(151,-4){$\hat l$}
\end{picture}
\caption{A comparison between the log sampling scheme and the scheme associated with the virtual frame $\{\hat\psi_{\hat l}\}_{\hat l}$. The index $\hat l$ is on  the $x$-axis (with an abuse of notation, we added the negative frequencies with negative indexes, as it is usual). The $y$-axis represents frequencies.}
\label{fig:log}
\end{figure}

According to our result, the measurements must be  chosen at random from $\{1,\dots,N\}$ with probabilities 
\[
\nu_l \propto  \lceil C_1^2 N /l\rceil,\qquad l=1,\dots,N.
\]
This nonuniform sampling scheme corresponds to a uniform sampling scheme for the virtual frame $\{\hat\psi_{\hat l}\}_{\hat l}$, in which each $\psi_l$ is repeated $\lceil C_1^2 N /l\rceil$ times and suitably normalized (see the proof of Corollary~\ref{cor:main}: note that the use of this uniform sampling scheme allows the choice $\u=1$). It is then natural to wonder how the frequencies $\{\hat k_{\hat l}\}_{\hat l}$ associated with $\{\hat\psi_{\hat l}\}_{\hat l}$ are arranged. Consider for simplicity the one-dimensional case with only the positive frequencies $k\in\N$, so that the ordering $\rho$ is simply the identity map. By construction, the frequency $\hat k_{\hat l}$  satisfies $\sum_{i=1}^{\hat k_{\hat l}} \lceil C_1^2 N /i\rceil \approx  \hat l$, and so $C_1^2 N\log \hat k_{\hat l}\approx \hat l$, which gives
\[
\hat k_{\hat l} \approx e^{\hat l/(C_1^2 N)}.
\]
This is the so-called \emph{log sampling scheme} \cite{2014-nonuniform,2016-gataric-poon} (up to a suitable scaling of the parameters), which is a 1D model for higher dimensional spiral trajectories; these are common in Magnetic Resonance Imaging. A comparison of the two sampling schemes is shown in Figure~\ref{fig:log}; the plot contains the negative frequencies as well (obtained simply by adding the positive frequencies with the signs changed). As expected, the log sampling scheme yields a smooth approximation of the sampling scheme $\{\hat k_{\hat l}\}_{\hat l}$. As far as the authors are aware, this is the first time that  theoretical support is given to the use of the log sampling scheme in CS. It is worth observing that the normalization of the elements $\hat\psi_{\hat l}$ by the square root of the number of repetitions corresponds to the standard normalization of weighted Fourier frames by the measures of Voronoi regions \cite{2017-adcock-gataric-hansen}.

\section{Applications}\label{sec:apps}

\subsection{Nonuniform Fourier sampling}
\label{sub:fourier}

The most classical compressed sensing problem formulated in the continuous setting is the recovery of a function $g\in L^2([0,1]^d)$ from Fourier samples
\[
(Ug)(k):= \hat g(k)=\int_{[0,1]^d} g(x) e^{-2\pi i k \cdot x}\,dx = \langle g,e^{2\pi i k \cdot}\rangle_{L^2([0,1]^d)},\qquad k\in\Omega,
\]
where $\Omega\subseteq\Z^d$ is a finite set of frequencies where the measurements are taken. Here $U$ is the discrete Fourier transform given by scalar products with the sinusoids $x\mapsto e^{2\pi i k \cdot x}$, which form an orthonormal basis of $ L^2([0,1]^d)$. If  $g$ is sparse with respect to a suitable orthonormal basis  with analysis operator $D$, this reconstruction problem fits in the framework discussed in the previous section, and $g$ may be recovered by $\ell^1$ minimization, provided that enough random measurements $\Omega$ are taken. The standard theory of compressed sensing may be applied in this case, since both $U$ and $D$ are unitary operators (see Example \ref{ex:isometry}).

In several applications (such as Magnetic Resonance Imaging, Computed Tomography, geophysical imaging, seismology and
electron microscopy), nonuniform Fourier sampling arises naturally, i.e.\ the frequencies are not taken uniformly in ${\Z^d}$. In this case, the operator $U$ fails to be an isometry, since the corresponding family of sinusoids may be only a frame, and not an orthonormal basis. The results discussed in the previous section may be directly applied to this case too.

Let us now give a quick overview of nonuniform Fourier frames; we follow \cite{2017-adcock-gataric-hansen}. For additional details, the reader is referred to \cite{2001-Christensen,2014-nonuniform} for the one-dimensional case, and to \cite{1966-beurling,2000-benedetto-wu,2012-olevskii-ulanovskii,2017-adcock-gataric-hansen} for the multi-dimensional case.

Let $\H$ be the space of square-integrable functions with support contained in a compact, convex and symmetric set $E\subseteq\R^d$, i.e.\ $\H=L^2(E)$. For $g\in\H$, we consider measurements of the form
\[
\hat g(k) = \int_E g(x) e^{-2\pi i k\cdot x}\,dx = \langle g,e_k\rangle_\H,\qquad k\in Z\subseteq\hat\R^d,
\]
namely scalar products with the sinusoids
\[
e_k(x) = e^{2\pi i k\cdot x},\qquad x\in E.
\]
Instead of considering the case when $Z$ is a cartesian grid of $\hat\R^d$ (here $\hat \R$ denotes the real line in the frequency domain), which gives rise to uniform Fourier sampling, we wish to give more general conditions on the set $Z$ so that $\{e_k\}_{k\in Z}$ is a frame of $\H$.

The first of these conditions requires that the samples are fine enough to capture all the frequency information in a given direction.
\begin{definition}[\cite{1966-beurling}]
We say that the sampling scheme $Z\subseteq\hat\R^d$ is \emph{$\delta$-dense} if
\[
\delta = \sup_{\hat y\in\hat\R^d}\inf_{k\in Z} |k-\hat y|_{E^\circ},
\]
where the norm $|\;|_{E^\circ}$ is given by
\[
|\hat y|_{E^\circ} = \inf\{a>0:x\cdot\hat y\le a \;\text{for every}\;x\in E\}.
\]
\end{definition}

The second condition limits the concentration of samples, in order to avoid large energies in small frequency regions.
\begin{definition}
We say that the sampling scheme $Z\subseteq\hat\R^d$ is \emph{separated} if there exists a constant $\eta>0$ such that
\[
\inf_{k_1,k_2\in Z, k_1\neq k_2} |k_1-k_2| \ge \eta >0.
\]
We say that $Z$ is \emph{relatively separated} if it is a finite union of separated sets.
\end{definition}

Under these conditions, the family of sinusoids $e_k$ with frequencies $k$ in $Z$ forms a frame for $L^2(E)$. 

\begin{proposition}[\cite{1966-beurling,2000-benedetto-wu,2012-olevskii-ulanovskii}]
Let $E\subseteq\R^d$ be a  compact, convex and symmetric set and take $\delta\in(0,1/4)$. If $Z\subseteq\hat\R^d$ is relatively separated and $\delta$-dense, then $\{e_k\}_{k\in Z}$ is a Fourier frame for $L^2(E)$.
\end{proposition}

Now, assuming that $\{ \varphi_j\}_{j\in \N}$ is a frame for $L^2(E)$, we can apply Theorem~\ref{thm:main} and Corollary~\ref{cor:main} to this setting. This would provide, to our knowledge, the first result about recovery of a sparse signal from nonuniform Fourier measurements via $\ell^1$ minimization. Even if the measurement frame is generally not tight, we can provide explicit bounds for the recovery of the wavelet coefficients from nonuniform Fourier samples.

Some numerical simulations related to this framework are presented in \cite[Example~5]{2016-gataric-poon}.

\subsection{Electrical impedance tomography}
\label{sub:eit}

EIT is an imaging technique in which one wants to determine the electrical conductivity $\sigma(x)$ inside a body $\O$ from boundary voltage and current measurements. It is a non-linear inverse boundary value problem whose mathematical formulation was presented for the first time by Calder\'on \cite{calderon1980}.

Let $\O \subset \R^d$, $d \geq 3$, be an open bounded domain with Lipschitz boundary and $\sigma \in L^{\infty}(\O)$, $\sigma (x) \geq \sigma_0 > 0$ for almost every $x \in \O$, be the electrical conductivity. Given a voltage $f \in H^{1/2}(\partial \O)$ on the boundary of the domain, the associated potential $u$ is the unique $H^1(\O)$ solution of the following Dirichlet problem for the conductivity equation:
\begin{equation}\label{eq:dircond}
\Div (\sigma \nabla u) = 0 \quad \text{in } \O, \qquad u = f \quad \text{on } \partial \O,
\end{equation}
where $H^\s$, $\s >0$, are the classical Sobolev spaces. The boundary current associated with the voltage $f$ is represented by the trace of the normal derivative of the potential $u$ on $\partial \O$. More precisely, we define the Dirichlet-to-Neumann (DN) map $\Lambda_\sigma \colon H^{1/2}(\partial \O) \to H^{-1/2}(\partial \O)$ as
\begin{equation}\label{def:DNs}
\Lambda_\sigma (f) = \left. \sigma \frac{\partial u}{\partial \nu}\right|_{\partial \O},
\end{equation}
where $\nu$ is the unit outer normal to $\partial \O$ and $u$ is the unique solution of \eqref{eq:dircond}. 

Calder\'on's inverse conductivity problem asks if it is possible to determine a conductivity $\sigma$ from the knowledge of its associated DN map $\Lambda_\sigma$. Positive answers to this question have been given since 1987 \cite{Sylvester1987,Novikov1988,Nachman1996}. 

If $\sigma$ is sufficiently smooth, the problem can be reduced to the so-called Gel'fand-Calder\'on problem for the Schr\"odinger equation,
\begin{equation}\label{eq:schr}
(-\Delta +q ) \tilde u = 0, \qquad q = \frac{\Delta \sqrt{\sigma}}{\sqrt{\sigma}},
\end{equation}
via the change of variables $u = \tilde u / \sqrt{\sigma}$ in \eqref{eq:dircond}. This inverse problem consists in the reconstruction of the potential $q$ from the knowledge of the DN map
\begin{equation}\label{def:DNq}
\Lambda_q \colon \tilde u|_{\partial \O} \mapsto \left.\frac{\partial \tilde u}{\partial \nu}\right|_{\partial \O}.
\end{equation}

One of the biggest open questions concerning inverse boundary value problems such as Calder\'on's or Gel'fand-Calder\'on's is the determination of a conductivity/potential from a finite number of boundary measurements. A priori assumptions on the unknown are needed in this case, and to the best of our knowledge the only result concerns piecewise constant coefficients with discontinuities on a single convex polygon \cite{friedman1989}. For conductivities/potentials belonging to some finite dimensional subspaces, an infinite number of measurements have always been a fundamental requirement to guarantee uniqueness and reconstruction. For instance, several works have studied the general piecewise constant case with infinitely many measurements \cite{2005-alessandrini-vessella,2011-beretta-francini}. In what follows, we will consider  finitely many measurements, and present a first result in this direction for the linearized  Gel'fand-Calder\'on problem, using the theory developed in this paper. A first result for the full nonlinear problem has been obtained very recently by the authors in \cite{alberti2018calder}.

In order to linearize the problem, we assume that $q = q_0 + \delta q$ where $q_0$ is known and $\delta q$ is small. Given two boundary voltages $f,g \in H^{1/2}(\partial \O)$ we have Alessandrini's identity \cite{alessandrini1988}:
\begin{equation*}
\langle g, (\Lambda_q - \Lambda_{q_0})f\rangle_{{ H^{\frac12}(\partial \O) \times  H^{-\frac12}(\partial \O)}} = \int_{\O}\delta q\, u_g u^0_f\,dx,
\end{equation*}
where $u_g$ (resp.\ $u^0_f$) solves the Schr\"odinger equation \eqref{eq:schr} with potential $q$ (resp.\ $q_0$) and Dirichlet data $g$ (resp.\ $f$). {The quantity on the left of this identity is known since $q_0$ is known and $\Lambda_q f$ is the boundary measurement corresponding to the chosen potential $f$ ($g$ should be seen as a test function).} Since for $\delta q \approx 0$ we have $u_g \approx u^0_g$, {the linearization consists in assuming} that we can measure the quantity $\int_{\O}\delta q \,u^0_g u^0_f\,dx$ for given $f,g$. {Focusing on the solutions themselves instead of on their boundary values, this inverse problem may be rephrased as follows.}

\begin{problem*}[Linearized Gel'fand-Calder\'on problem]
Given a finite number of scalar products of the form $\int_\O \delta q\, u_1 u_2\,dx$, where $u_1$ and $u_2$ are solutions of 
\begin{equation}\label{eq:u0}
(-\Delta +q_0 ) u_i = 0
\end{equation}
in $\O$, find $\delta q \in L^2(\O)$.
\end{problem*}

Without loss of generality, we can assume that $\O\subseteq\T^d$, where $\T=[0,1]$. Extend $\delta q$ by zero to $\H:=L^2(\T^d)$ and assume that $q_0 \in H^\s(\T^d)$ for some $\s > d/2$. In the rest of this subsection, with an abuse of notation, several different positive constants depending only on $d$, $\s$ and $\|q_0\|_{H^\s(\T^d)}$ will be denoted by the same  letter $c$. By a classical uniqueness result for the Calder\'on problem \cite{Sylvester1987} we have that for every $k \in \Z^d$
and $t\ge c$ we can construct solutions $u_i^{k,t}$ of \eqref{eq:u0} in $\T^d$ of the form
\[
u_i^{k,t}(x) = e^{\zeta_i^{k,t}\cdot x}(1+r(x,\zeta_i^{k,t})),\qquad x\in\T^d,
\]
where $\zeta_i^{k,t}\in\C^d$ are such that  $\zeta_1^{k,t}+ \zeta_2^{k,t} =- 2\pi i k$ and

\begin{equation}\label{eq:est_r}
\|r(x,\zeta_i^{k,t})\|_{H^\s(\T^d)} \leq \frac{c}{t}, \qquad i=1,2.
\end{equation}
These solutions $u_i^{k,t}$ are known as exponentially growing solutions, Faddeev-type solutions \cite{faddeev1965} or complex geometrical optics (CGO) solutions.

We need to consider an ordering of $\Z^d$, namely a bijective map $\rho\colon\N\to \Z^d$, $l\mapsto k_l$.
For each $k \in \Z^d$ fix {$t_k\ge c$} and define the measurement operator $U_{GC}\colon \H \to \ell^2(\N)$ by
\begin{equation*}\label{def:U_EIT}
U_{GC} (\delta q)= (\langle \delta q, \psi_l \rangle)_l,\qquad \overline{\psi_l} = u_1^{k_l,t_{k_l}} u_2^{k_l,t_{k_l}}.
\end{equation*}
We call the family $\{ \psi_l\}_{l \in \N}$ a CGO frame (see Lemma~\ref{lem:eit} below). Using the same ordering of $\Z^d$, we define the discrete Fourier transform $F\colon \H \to \ell^2(\N)$ by
\begin{equation*}\label{def:DFT}
F(\delta q)=(\langle \delta q, e_{k_l} \rangle)_l,
\end{equation*}
where $e_k(x)=e^{2\pi ik\cdot x}$.

We can now state the following consequence of Corollary~\ref{cor:main} for the linearized EIT problem.
\begin{cor}\label{cor:eit}
Let $\{\varphi_j\}_j$ be an orthonormal basis of $L^2(\T^d)$ and $D\colon  L^2(\T^d)\to\ell^2(\N)$ be its analysis operator. Let $M,s\in \N$ and $\omega \geq 1$ be such that  $3\le s\le M$. Assume that $N\ge s$ satisfies the balancing property with respect to  $F$, $D$, $M$ and $s$ with the right-hand side of the inequalities \eqref{bp1} and \eqref{bp2} divided by 2. Assume that there exists $C_1>0$ such that
\begin{equation}\label{hyp:asyn}
\sup_{j\in \N} |\langle \varphi_j,e_{k_l}\rangle| \le \frac{C_1}{\sqrt{l}}\qquad
\text{and}\qquad
 \sup_{l \in\N} |\langle \varphi_j,e_{k_l}\rangle| \le \frac{C_1}{\sqrt{j}}
\end{equation}
for every $l\le N$ and $j>M$.
 Let $\{\psi_l\}_l$ be the CGO frame constructed with {$t_{k} \geq c\sqrt{N} (|k|^\s+1)$ for every $k\in\Z^d$}. Assume
\begin{equation*}\label{eq:EIT}
m \geq C \, \omega^2 s \log^2 N,
\end{equation*}
where $C>0$ is a constant depending only on $C_1$.

Sample $m$ indices $l_1, \ldots, l_m$ indipendently from $\{1,\ldots,N\}$  according to the probability distribution
$ \nu_l = C_N \lceil (C_1+1)^2 N/l\rceil$, where $C_N = \left(\sum_{l=1}^N \lceil (C_1+1)^2 N /l \rceil \right)^{-1}$, and set $\Omega = \{l_1,\ldots,l_m\}$.

Let $\delta q \in L^2(\T^d)$ and $\eta \in \C^m$ be such that $\norm{\eta}_w\le\epsilon$ for some $\epsilon\ge 0$, where $\| \eta\|_w^2 = \sum_{i=1}^m \frac{|\eta_i|^2}{\lceil (C_1+1)^2 N /l_i \rceil}$. Let $\zeta_i = \int_{\T^d} \delta q\, \overline{\psi_{l_i}}\,dx+\eta_i$ be the known noisy measurements.
 Let $g \in L^2(\T^d)$ be a minimizer of
\begin{align*}\label{min:eit}
\inf_{\substack{
 g \in L^2(\T^d) \\  D  g \in \ell^1(\N)
}} \Vert D g \Vert_1 \quad \text{subject to } \norm{\left(\begingroup\textstyle\int\endgroup_{\T^d} g\, \overline{\psi_{l_i}}\,dx - \zeta_i\right)_i}_w \le\epsilon.
\end{align*}
Then, with probability exceeding $1-e^{-\omega}$, we have
\[
\norm{g - \delta q}_{L^2(\T^d)} \leq 80\sigma_{s,M}(D(\delta q)) + C' \sqrt{ \frac{\omega s N\log N}{m}} \epsilon,
\]
where $C'$ is a universal constant.
\end{cor}

Corollary \ref{cor:eit} provides a first general recipe to recover or approximate a sparse or compressible conductivity from a small number of linearized EIT measurements. The assumption that the sparsifying basis $\{\varphi_j\}_{j \in \N}$ and the Fourier basis $\{e_{k_l}\}_{l \in \N}$ must be asymptotically incoherent is not restrictive: as already mentioned in $\S$\ref{sub:wavelets}, a large class of wavelet bases satisfy  \eqref{hyp:asyn} \cite{2016arXiv161007497J}. Note that in the 1D Fourier-Wavelet case we have $N = O(M \log M)$ (Remark~\ref{rem:balancing2}). It would be very interesting to test this algorithm numerically: this is left for future work.

 Note that for the incoherence and balancing property we  made assumptions only on the Fourier basis and not directly on the CGO frame $\{ \psi_l\}_l$. This is possible thanks to the following lemma, which also shows that $U_{GC}$ is an invertible operator with bounded inverse, provided that the $t_k$s are chosen big enough.

\begin{lemma}\label{lem:eit}
{There exists $c'>0$ depending only on $d$, $\s$ and $\|q_0\|_{H^\s(\T^d)}$ such that if $t_k\ge c'\lambda(|k|^\s+1)$ for every $k\in\Z^d$ and for some $\lambda\ge 2$ then the operator $U_{GC}$ is bounded and invertible and
\begin{align}\label{est:ugc}
& \|U_{GC}\|_{\H \to \ell^2(\N)} \leq  3 /2,\qquad&&\|U_{GC}^{-1}\|_{\ell^2(\N)\to \H} \leq 2,\\ \label{est:ugcF}
&\|U_{GC}-F\|_{\H \to \ell^2(\N)} \leq {1}/{\lambda}, \qquad &&\|U_{GC}^{-1}-F^\ast\|_{\ell^2(\N) \to \H} \leq 1/{\lambda},\\ \label{est:rl1}
&\sup_{l \in \N} \| F (r(\,\cdot\,,\zeta_i^{k_l,t_{k_l}}))\|_{\ell^1(\N)} \le {1}/{\lambda}, && \text{for } i = 1,2.
\end{align}
}
\end{lemma}
\begin{proof}
Since $\overline{\psi_l}=e^{- 2\pi ik_l\cdot x}(1+r(x,\zeta_1^{k_l,t_{k_l}})) (1+r(x,\zeta_2^{k_l,t_{k_l}})) $, setting $r_i^l=r(\,\cdot\,,\zeta_i^{k_l,t_{k_l}})$ we have
\begin{equation*}
\norm{e_{k_l}-\psi_l}_{L^2(\T^d)} \le
\|r_1^l\|_{L^2(\T^d)}+\|r_2^l\|_{L^2(\T^d)}
+\|r_1^l\|_{L^{\infty}(\T^d)}\|r_2^l\|_{L^2(\T^d)}
\le \frac{c}{|t_{k_l}|},
\end{equation*}
where we used estimate \eqref{eq:est_r} and the Sobolev embedding $H^\s(\T^d) \hookrightarrow L^{\infty}(\T^d)$ for $\s > d/2$. {This implies $|((U_{GC}-F)\delta q)_l|\leq  \|\delta q\|_{L^2(\T^d)}\frac{c}{c'\lambda(|k_l|^\s+1)}$, so that
\begin{equation*}
\|U_{GC}-F\|_{\H \to \ell^2(\N)}
 \leq \frac{c}{c'\lambda} \Biggl(\sum_{k\in\Z^d}\frac{1}{(|k|^\s+1)^2}\Biggr)^\frac12
\le  \frac{c}{c'\lambda}
\end{equation*}
(note that the series is finite as $2\s > d$). Choosing $c'\ge c$   immediately yields $\|U_{GC}-F\|\le \frac{1}{\lambda}\le\frac12$. The first part of \eqref{est:ugcF} follows. Hence $\norm{U_{GC}}\le\norm{U_{GC}-F} + \norm{F}\le \frac32$, since $F$ is an isometry.} Moreover, we have the Neumann series expansion
\begin{equation*}
U_{GC}^{-1} = F^{-1}\sum_{k=0}^{+\infty}(-1)^k\left( (U_{GC}-F)F^{-1}\right)^k,
\end{equation*}
and so $\norm{U_{GC}^{-1}}\le 2$, as desired. In addition, the latter identity readily gives the second part of \eqref{est:ugcF}, since $F^{-1} = F^\ast$. The last estimate \eqref{est:rl1} can be proven as follows. {We readily derive
\[
\begin{split}
\| F r^l_i\|_1 &= \sum_{h \in \N}|(F r^l_i)_h| (1+|k_h|^2)^{\frac{\s}{2}}(1+|k_h|^2)^{-\frac{\s}{2}}\\ & \leq \sqrt{\sum_{h \in \N}|(F r^l_i)_h|^2 (1+|k_h|^2)^\s}\sqrt{\sum_{h \in \N}(1+|k_h|^2)^{-\s}}.
\end{split}
\]
Since $\s > d/2$, the series on the right is convergent. Hence, by definition of the Sobolev norm we obtain
\[
\| F r^l_i\|_1  \leq c \|r_i^l\|_{H^\s(\T^d)} \le \frac{c}{ c'\lambda(|k_l|^\s+1)} \le \frac{c}{ \lambda c'}\le \frac{1}{\lambda},
\]
provided that $c'\ge c$.} This finishes the proof of the lemma.
\end{proof}
We are now in a position to prove Corollary~\ref{cor:eit}.
\begin{proof}[Proof of Corollary~\ref{cor:eit}]
We want to apply Corollary~\ref{cor:main} to $U_{GC}$ and $D$. Set $c= 84c'$, where $c'$ is given by Lemma~\ref{lem:eit}, so that $\lambda = 84\sqrt{N}$.

First note that $\k_1 = 4$ by \eqref{est:ugc} and $\k_2=1$ since $\{\varphi_j\}_j$ is an orthonormal basis. We begin by showing that $N$ satisfies the balancing property (with the original right hand side in the definition) with respect to  $U_{GC}$, $D$, $M$ and $s$. Consider condition \eqref{bp2} (\eqref{bp1} can be shown by using the same argument): by the triangle inequality we have
\begin{multline*}
\|P_\Delta^\perp D^{-*} P_{\W}^\perp U_{GC}^* P_N U_{GC}^{-*} P_{\W}\|_{\H \to \ell^{\infty}}
\le
\|P_\Delta^\perp D^{-*} P_{\W}^\perp (U_{GC}^*-F^*) P_N U_{GC}^{-*} P_{\W}\|_{\H \to \ell^{\infty}}\\
+ \|P_\Delta^\perp D^{-*} P_{\W}^\perp F^* P_N (U_{GC}^{-*}-F) P_{\W}\|_{\H \to \ell^{\infty}}
+ \|P_\Delta^\perp D^{-*} P_{\W}^\perp F^* P_N F P_{\W}\|_{\H \to \ell^{\infty}},
\end{multline*}
and so by \eqref{est:ugc}-\eqref{est:ugcF}  and  since $N$ satisfies the balancing property with respect to  $F$, $D$, $M$ and $s$ with the bounds divided by 2 we obtain
\[
\|P_\Delta^\perp D^{-*} P_{\W}^\perp U_{GC}^* P_N U_{GC}^{-*} P_{\W}\|_{\H \to \ell^{\infty}} 
\le \frac{2}{\lambda} + \frac{1}{\lambda} + \frac{1}{28\sqrt{s}} \le \frac{1}{14\sqrt{s}},
\]
since $\lambda \geq 84 \sqrt{s}$.

Next, we show that  estimate \eqref{def:async} with weights $w_l=\frac{C_1+1}{\sqrt{l}}$ is valid for the CGO frame:
\begin{align*}
&\sup_{j\in \N} \max\{ |\langle \varphi_j,\psi_{l}\rangle_\H|, |\langle \varphi_j,\tilde \psi_l \rangle_\H| \} \le \frac{C_1+1}{\sqrt{l}},\qquad l=1,\dots,N,
\end{align*}
This inequality is readily obtained from the first bound in \eqref{hyp:asyn} and the estimates \eqref{est:ugcF}, since $\lambda \geq \sqrt{N}$.

In order to simplify the term $\tilde M$ in the estimate for the number of measurements $m$, we now show the following decay of the coherence:
\begin{equation}\label{eq:estj}
\sup_{l \le N} |\langle \varphi_j,\psi_{l}\rangle_\H| \leq \frac{2C_1}{\sqrt{j}} ,\qquad j>M.
\end{equation}
 Write $\psi_l = e_{k_l}r_l$, 
where  $\overline r_l = 1 + r_1^l+r_2^l+r_1^lr_2^l$ and $r_i^l = r(\,\cdot\,,\zeta_i^{k_l,t_{k_l}})$. Expand $r_l$ in Fourier series $r_l = \sum_{h \in \N} (F r_l)_h e_{k_h}$, so that $e_{k_l}r_l = \sum_{h \in \N} (F r_l)_h e_{k_h+k_l}$. Then we have
\[
 |\langle \varphi_j,\psi_{l}\rangle_\H|=|\langle \varphi_j,e_{k_l}r_l \rangle_\H|\leq  \| F r_l\|_{\ell^1(\N)} \sup_{h \in \N}  |\langle \varphi_j,e_{k_l+k_h} \rangle_\H| \leq \frac{C_1  \| F r_l\|_{\ell^1(\N)}}{\sqrt{j}},
\]
using \eqref{hyp:asyn}. By Young's inequality for discrete convolution and \eqref{est:rl1} we have
\[
\begin{split}
 \| F r_l\|_{\ell^1(\N)} &\le  \| F 1\|_{\ell^1(\N)} + \| F \overline{r^l_1}\|_{\ell^1(\N)} +  \| F \overline{r^l_2}\|_{\ell^1(\N)} + \| (F \overline{r^1_l})* (F\overline{r^2_l})\|_{\ell^1(\N)} \\
&\le  1 + \frac{1}{\lambda}+ \frac{1}{\lambda} + \frac{1}{\lambda^2} \\
&\le 2.
\end{split}
\]
Combining the last two inequalities we obtain \eqref{eq:estj}. By Remark~\ref{rem:M}, this implies $\tilde M ( \alpha) \leq 16 C_1^2 \frac{N}{\alpha^2}$.

We can now  apply Corollary \ref{cor:main}, and the result follows.
\end{proof}

Let us make some concluding remarks about this inverse problem.
We have chosen the functions $u_i^{k,t}$ from \cite{Sylvester1987} for the sake of simplicity. Other families of functions with similar decay properties might be used as well, leading to similar results as Lemma \ref{lem:eit}, with lower regularity assumptions on the coefficients to be recovered.

In two dimensions it is unclear if results such as Lemma \ref{lem:eit} could hold: for the linearized Calder\'on problem we cannot use CGO solutions to approximate the Fourier transform as in higher dimensions. For the linearized Gel'fand-Calder\'on problem one could use the Bukhgeim approach \cite{bukhgeim2008} to recover pointwise values of a potential via stationary phase type techniques.

More generally, the results of this subsection may be applied to a large class of linearized inverse boundary value problems for which we have families of CGO solutions with good decay properties: inverse problems for the Helmoltz equation, the elasticity system and Maxwell's equations, for instance.

\subsection{An inverse problem for the wave equation}
\label{sub:wave}

Our main result can also be applied to another linear infinite dimensional inverse problem, the observability problem for the wave equation \cite{1986-ho,1988-lions,2000-lasiecka-triggiani,2012-ervedoza-zuazua,alberti-capdeboscq-2016}. This is a classical inverse problem, and consists in the reconstruction of the initial source of the wave equation from boundary measurements of the solution. In addition to the direct link with control theory, this inverse problem appears in the formulation of thermoacoustic and photoacoustic tomography in a bounded domain \cite{2010-ammari-bossy-jugnon-kang,2013-kunyansky-holman-cox,2015-acosta-montalto,2015-holman-kunyansky,2016-chervova-oksanen} (for the free-space formulation, see \cite{2015-kuchment-kunyansky}).

Let $d\ge 2$ and $\O\subseteq\R^d$ be a bounded smooth domain. We consider the following initial value problem for the wave equation\footnote{For simplicity, we consider the case of constant sound speed (normalized to $1$), but this analysis may be generalized to the case of a spatially varying sound speed $c$. Similarly, considering a non-homogeneous initial condition for $\partial_t p$ would not add any substantial complications.}
\begin{equation}\label{def:wavefor}
\left\{
\begin{array}{ll}
\partial_{tt}p-\Delta p = 0 & \text{in } (0,T) \times \O,\\
p(0,\cdot) = f & \text{in } \O,\\
\partial_t p(0,\cdot) = 0 &  \text{in } \O,\\
p= 0 &  \text{on } (0,T) \times \partial \O,
\end{array}\right.
\end{equation}
where $T>0$ and $f\in H^1_0(\O):=\{u\in H^1(\O):u=0\;\text{on $\partial\O$}\}$ is the unknown initial condition. The above problem admits a unique weak solution $p\in C([0,T];H^1_0(\O)$) (see \cite[Section~7.2]{2010-evans} and \cite[Theorem~10.14]{2011-brezis}). The inverse problem of interest may be formulated as follows.

\begin{problem*}[Observability of the wave equation]
Supposing that the trace of the normal derivative $\partial_\nu p$ is measured on an open subset $\Gamma$ of $\partial \O$ for all $t \in (0,T)$, where $\nu$ is the exterior unit normal to $\partial\Omega$, find the initial condition $f$ in $\O$.
\end{problem*}

Observe that the forward problem is always well-posed:  by an inequality of Rellich's, the measurement operator
\[
V\colon H^1_0(\O) \to L^2((0,T)\times\Gamma),\qquad f\mapsto \partial_\nu p,
\]
where  $p$ is the solution of \eqref{def:wavefor}, is well-defined and bounded \cite[(1.20)]{1988-lions}.

In order to apply our techniques to the inverse problem we need more than continuity, namely injectivity and bounded invertibility of the map $V$. In this case, $f$ is uniquely and stably determined by the boundary data $Vf=\partial_\nu p$ on $(0,T) \times \Gamma$. This solves the above-mentioned inverse problem when we can perfectly measure $\partial_\nu p$ on the whole $(0,T)\times \Gamma$.

There is a wide literature concerning assumptions on $\Gamma$ and $T$ that guarantee the invertibility of $V$ (see \cite{1992-blr} and references therein). Here we only mention a sufficient condition by Ho \cite{1986-ho} and J.\ L.\ Lions \cite{1988-lions} (see also \cite[$\S$5.3.4]{2012-ervedoza-zuazua} and \cite[Theorem~2.8]{alberti-capdeboscq-2016}):
if 
$
\{x \in \partial \O : (x-x_0)\cdot \nu > 0 \} \subseteq \Gamma
$
for some $x_0 \in \R^d$ 
and
$
T >  2 \sup_{x\in \O}|x-x_0|,
$
then $V$ is invertible with bounded inverse.  In the following, we shall assume that $V$ is invertible with bounded inverse.

In order to let compressed sensing  come into play, we will make use of the following identity, which follows by a simple integration by parts \cite[Corollary~2.13]{alberti-capdeboscq-2016}. For every  $v \in L^2((0,T)\times \Gamma)$, we have
\begin{equation}\label{eq:idwave}
(Vf,\bar v)_{L^2((0,T)\times \Gamma)}=\int_{(0,T) \times \Gamma}\partial_\nu p\,  v \,dt d \sigma =\langle  \partial_t {U_v(0,\cdot)},f\rangle_{H^{-1}(\O), H^1_0(\O)},
\end{equation}
where $U_v \in C\left([0,T]; L^2(\O)\right) \cap C^1\left( [0,T] ; H^{-1}(\O)\right)$ is the solution of 
\begin{equation}\label{def:wavetrans}
\left\{
\begin{array}{ll}
\partial_{tt}U_v-\Delta U_v = 0 & \text{in } (0,T) \times \O,\\
U_v(T,\cdot) = 0 & \text{in } \O,\\
\partial_t U_v(T,\cdot) = 0 &  \text{in } \O,\\
U_v= \chi_\Gamma v &  \text{on } (0,T) \times \partial \O,
\end{array}\right.
\end{equation}
which is defined in the sense of  transposition \cite{1988-lions-a,1988-lions}, 
where $\chi_\Gamma$ is the characteristic function of $\Gamma$ and $H^{-1}(\O)$ is the dual of $H^1_0(\O)$. 
Identity \eqref{eq:idwave} shows that we can use the dual solution $U_v$ to \textit{probe} the unknown $f$: we measure different moments of $f$ by varying $v$. 

Since observability is equivalent to exact controllability \cite{1988-lions-a,1988-lions}, we have that for every $h \in H^{-1}(\O)$ there exists $v_h \in L^2\left( (0,T) \times \Gamma \right)$ such that $\partial_t U_{v_h} (0,\cdot) = h$.  The control $v_h$ can be explicitly constructed via an optimization problem. By \eqref{eq:idwave} we obtain:
\begin{equation*}
(V f, \overline{v_h})_{L^2((0,T)\times \Gamma)} = \langle h, f \rangle_{H^{-1}(\O), H^1_0(\O)}.
\end{equation*}
Let $-\Delta\colon H^1_0(\O)\to  H^{-1}(\O)$ be the Dirichlet Laplacian. By definition we have $\langle -\Delta\overline{\psi},g\rangle_{H^{-1}(\O), H^1_0(\O)} = (g,\psi)_{H^1_0(\O)}$ (the scalar product in $H^1_0(\O)$ is defined by $ (g,\psi)_{H^1_0(\O)}=\int_\O\nabla g\cdot \overline{\nabla \psi}\,dx$).
 Inserting this expression into the above identity yields
\begin{equation}\label{eq:idwave2}
(V f, {v_{-\Delta\psi}})_{L^2((0,T)\times \Gamma)}  =(V f, \overline{v_{-\Delta\overline\psi}})_{L^2((0,T)\times \Gamma)}  = (f, \psi )_{H^1_0(\O)}, \quad \psi \in H^1_0(\O).
\end{equation}

Let $\{ \psi_l\}_{l\in\N}$ be a frame of $H^1_0(\O)$ and $\{v_l\}_{l\in\N}$ be a family  of $L^2((0,T)\times \Gamma)$ such that
\begin{equation}\label{eq:v-psi}
v_l = v_{-\Delta\psi_l}\;  \iff \;\psi_{l} = (-\Delta)^{-1} \partial_t U_{{v_l}}(0,\cdot).
\end{equation}
These relations show that one may first choose the frame $\{ \psi_l\}_l$ and then construct the related family $\{v_l\}_l$, or viceversa. Define the measurement operator
\begin{align*}
U_{obs} : H^1_0(&\O) \to \ell^2(\N),\qquad f \mapsto  \bigl(( f, \psi_l )_{H^1_0(\O)}\bigr)_l,
\end{align*}
which can be measured, thanks to \eqref{eq:idwave2}. Then, representing $f$ in another frame $\{ \varphi_j \}_j$ of $H^1_0(\O)$ we can apply Theorem \ref{thm:main} (or Corollary \ref{cor:main}) to this setting, provided that $\{ \psi_l\}_l$ and  $\{ \varphi_j \}_j$ are incoherent (or asymptotically incoherent). Therefore, via $\ell^1$ minimization we can reconstruct $f$ from the partial measurements $\{( f, \psi_l )_{H^1_0(\O)}\}_{l\in\Omega}$, for some subsampling subset $\Omega\subseteq\N$, provided that $f$ is sparse with respect to   $\{ \varphi_j \}_j$.

Note that, in order to measure $( f, \psi_l )_{H^1_0(\O)}=(V f, v_l)_{L^2((0,T)\times \Gamma)}$, in principle we might need to know $Vf$ on the whole $(0,T) \times \Gamma$. The subsampling procedure would then become useless. In order to overcome this issue, one has to choose the functions $v_l$ in such a way that the computation of each $(V f, v_l)_{L^2((0,T)\times \Gamma)}$ only requires a partial knowledge of $Vf$. For instance, one could choose compactly supported functions $v_l$'s in order to sample subsets of $(0,T) \times \Gamma$: this would correspond to having sensors only on particular locations of the boundary at specific times. Similarly, scalar products with slowly varying $v_l$s would correspond to local averages of $Vf$, which may be obtained with integrating area or line detectors \cite{2005-TAT-lines,2015-CS-PAT,2016-CS-PAT}. More general $v_l$'s are considered in \cite{PAT}.

In summary, the challenge is to construct families $\{\varphi_j\}_j,\{\psi_l\}_l\subseteq H^1_0(\O)$ and $\{v_l\}_l\subseteq L^2((0,T)\times \Gamma)$ such that:
\begin{itemize}
\itemsep0em
\item $\{ \psi_l\}_l$ and  $\{\varphi_j\}_j$ are frames of $H^1_0(\O)$;
\item $\{ \psi_l\}_l$ and $\{v_l\}_l$ are related via \eqref{eq:v-psi}, which involves the solution of the PDE \eqref{def:wavetrans};
\item $\{ \psi_l\}_l$ and $\{\varphi_j\}_j$ are  incoherent (or asymptotically incoherent);
\item and each scalar product $(V f, v_l)_{L^2((0,T)\times \Gamma)}$ may be computed with partial measurements of $Vf$.
\end{itemize}
A detailed analysis of these issues goes beyond the scope of this paper, and is a very interesting direction for future work, at the interface of applied harmonic analysis and PDE theory.

\section{Proof of Theorem~\ref{thm:main}}\label{sec:proof}

The aim of this section is to prove Theorem~\ref{thm:main}.

\subsection{Concentration inequalities}
Certain large deviation bounds for sums of vector and matrix valued random variables are required to prove some of the key results. Inspired by the paper of Kueng and Gross \cite{KG} we use Bernstein inequalities instead of applying Talagrand as done by Adcock and Hansen \cite{AH}. We give a particular vector inequality not depending on the dimension taken from \cite[Proposition~7]{KG} which originally appears in \cite[Chapter 6.3, Eqn.~(6.12)]{LT}  with a direct proof in \cite{G}.

\begin{lemma}[Vector Bernstein inequality]
\label{lem:vectorbern}
Let $\{X_l\}\subset \mathbb{C}^d$ be a finite sequence of independent random vectors. Suppose that $\E[X_l]=0$, $\Vert X_l \Vert_2 \leq B$ almost surely and  $\sum_l \E \left[ \Vert X_l \Vert_2^2\right]\le \sigma^2$ for some $B,\sigma>0$. Then for all $0 \leq t \leq \frac{\sigma^2}{B}$
\begin{equation}\label{vector}
\P \left( \left\Vert \sum_l X_l \right\Vert_2 \ge t\right) \leq \exp\left( \frac{-t^2}{8\sigma^2} + \frac{1}{4}\right).
\end{equation}
\end{lemma}

The matrix deviation estimate that we use is due to Tropp \cite[Theorem~1.6]{Tr}.

\begin{lemma}[Matrix Bernstein inequality]\label{lem:matrixbern}
Consider a finite sequence $\{X_l\} \subset \mathbb{C}^{d\times d}$ of independent random matrices. Assume that each random matrix satisfies $\E [X_l]=0$ and $\Vert X_l \Vert \leq B$ almost surely, where $\| \cdot \|$ stands for the spectral norm, i.e.\ the natural norm induced by $\| \cdot \|_2$. Define
\[
\sigma^2 := \max \left\{\left\Vert \sum_l \E\left(X_l X_l^* \right) \right\Vert, \left\Vert \sum_l\E\left(X_l^{*}X_l \right)\right\Vert\right\}.
\]
Then for all $t\geq 0$
\begin{equation}\label{matrix}
\P \left(\left\Vert \sum_l X_l \right\Vert \geq t \right) \leq 2d \exp\left(\frac{-t^2/2}{\sigma^2 + Bt/3} \right).
\end{equation}
\end{lemma}

\subsection{Six useful estimates}\label{sub:four}

Our proofs rely on several estimates. We provide them below, following mostly \cite{AH,KG,Poon2015}, and using a structure similar to \cite{CP}. In order to avoid repetitions and enhance clarity, we summarize here the assumptions we make throughout this subsection:
\begin{itemize}
\itemsep0em
\item Assume that Hypothesis~\ref{hp1} holds true, and let $U$ and $D$ denote the corresponding analysis operators with index sets $L$ and $J$, respectively, satisfying the bounds given in \eqref{eq:boundk};
\item Let $M\in {J}$ and $\Delta\subseteq \{1,\dots,M\}$ satisfy $|\Delta|\ge 2$, and set $\W=R(D^{\ast} P_{\Delta})+R(D^{-1} P_{\Delta})$;
\item Let $N\in {L}$ satisfy the balancing property with respect to {$U$, $D$,} $M$ and $|\Delta|$;
\item For a fixed $\theta\in (0,1 ]$, let $\{N, \ldots, 1\} \supseteq \Omega \sim Ber (\theta)$, i.e.\ 
\[
\Omega=\{l\in\{1,\dots,N\}:\delta_l=1\},
\]
where $\{\delta_l\}_{l=1}^N$ are Bernoulli variables with $\P (\delta_l = 1) = \theta$;
\item Set $E_\Omega = U^{*}P_\Omega \U$.
\end{itemize}

The first estimate reads as follows.

\begin{proposition}\label{prop9.1}
For $g \in \H$ and $t \ge \frac{1}{7\sqrt{|\Delta|\k_2}}$ we have
\begin{multline*}
 \P \left( \Vert \theta^{-1} P_\Delta^{\perp}D^{-*} P^\perp_{\W} E_{\Omega} P_{\W} g \Vert_{\ell^\infty {(J)}} > t \Vert g \Vert_{\H} \right)
 \\ \le 2\tilde M(\tfrac {t\theta}{2}) \exp \left(\frac{-t^2 \theta\mu^{-2}}{8\k_1 B_\Delta\bigl(B_\Delta + \frac{\eta_\Delta\sqrt{2|\Delta|}t}{6}\bigr)} \right).
\end{multline*}
\end{proposition}

\begin{proof}
 Without loss of generality we may assume that $\Vert g \Vert_{\H}=1$. We shall need the following inequality:
\begin{align}
\label{eq:withB2}
&|\langle U P_{\W}^\perp D^{-1} e_j, e_l \rangle |\le 
|\langle e_j, D^{-*} P_{\W}^\perp D^{*}D^{-*} U^\ast e_l \rangle |
\le B_\Delta \|D^{-*} U^\ast e_l \|_\infty
\le B_\Delta \mu,
\end{align}
where $B_\Delta$ and $\mu$ are defined in \eqref{eq:defB} and in Definition~\ref{def:mu}, respectively.

Since $\sum\limits_{l=1}^N  e_l e_l^{*} = P_N$ and $\sum\limits_{l=1}^N \delta_l e_l e_l^{*} = P_\Omega$ we have
\begin{equation}\label{eq:Yk}
\theta^{-1} P_\Delta^{\perp}D^{-*} P^\perp_{\W} E_{\Omega} P_{\W} g = \sum_{l=1}^N Y_l + P_\Delta^\perp D^{-*} P_{\W}^\perp U^* P_N \U P_{\W} g,
\end{equation}
where $Y_l = \theta^{-1} P_\Delta^\perp D^{-*} P^\perp_{\W} U^* (\delta_l-\theta)e_l e_l^* \U P_{\W} g$. For $j\in J$  we define the random variable $X_l^j = \langle Y_l, e_j \rangle$. By the balancing property \eqref{bp2} 
we have
\[
\P \left( \Vert \theta^{-1} P_\Delta^{\perp}D^{-*} P^\perp_{\W} E_{\Omega} P_{\W} g \Vert_{\infty } > t \Vert g \Vert_{\H} \right)
 \le \P \left( \left\|  \sum_{l=1}^N Y_l \right\|_\infty > \frac{t}{2}\right).
\]

Let us estimate this quantity by studying the random variables $X^j_l$ via Lemma~\ref{lem:matrixbern} with $d=1$. In order to do that, first observe that since $\E(\delta_l)= \theta$, then $\E(X_l^j)= 0$. We next study the upper bounds on $\E\bigl(|X_l^j|^2 \bigr)$ and $|X_l^j|$ for $l=1, \ldots, N$. 

On the one hand, by \eqref{eq:withB2} we have
\begin{equation*}
\begin{split}
|\langle D^{-*} P_{\W}^\perp  U^{\ast} e_l e_l^{*} \U P_{\W} g, e_j \rangle| & = |\langle \U P_{\W} g, e_l \rangle| |\langle U P_{\W}^\perp D^{-1} e_j, e_l \rangle | \\ &\leq \mu B_\Delta |\langle \U P_{\W} g, e_l \rangle|,
\end{split}
\end{equation*}
so that  $\E \left[ (\delta_l- \theta)^2\right]= \theta(1-\theta)$ implies  for $j\in J$
\begin{align*}
 \E \bigl( |X_l^j|^2 \bigr)  & = \theta^{-2} \E \left( (\delta_l -\theta)^{2}\left|\langle P_\Delta^{\perp} D^{-*} P_{\W}^\perp  U^{\ast} e_l e_l^{*} \U P_{\W} g, e_j \rangle \right|^2 \right)\\
 & \leq \theta^{-1}(1-\theta)\mu^2 B_\Delta^2 \left|\langle \U P_{\W} g, e_l \rangle\right|^2.
\end{align*}
Therefore, since $\Vert \U \Vert \leq \sqrt{\k_1}$, we deduce that
\[
\sum_{l =1}^N \E \bigl( |X_l^j|^2 \bigr)  \leq \theta^{-1}(1-\theta)\mu^2 B_\Delta^2 \Vert \U P_{\W} g \Vert_2^2 \leq \theta^{-1} \mu^2 B_\Delta^2 \k_1 = : \sigma^2.
\]

On the other hand, setting $h_i= \frac{P_{\W_i} U^{-1}e_l}{\norm{P_{\W_i} U^{-1}e_l}_\H}\in\W_i$ for $i=0,1$, where $\W_i = R(D_i^\ast P_\Delta)$, $D_0 =D$ and $D_1 = D^{-\ast}$, we readily derive
\begin{equation*}
\norm{P_{\W_i} U^{-1}e_l}_\H =| \langle h_i,  U^{-1}e_l\rangle| 
= | \langle D_i h_i, D_i^{-*}U^{-1}e_l\rangle| 
\le \mu\norm{D_i h_i}_1
\le \mu\eta_\Delta\sqrt{|\Delta|},
\end{equation*}
where $\eta_\Delta$ is given by \eqref{eq:eta}. This yields the estimate
\begin{equation}\label{eq:estpwk}
\norm{P_\W U^{-1}e_l}_\H^2 \leq \norm{P_{\W_0} U^{-1}e_l}_\H^2 + \norm{P_{\W_1} U^{-1}e_l}_\H^2 \leq 2 \mu^2 \eta_\Delta^2|\Delta|.
\end{equation}
For later use, note that we analogously have
\begin{equation}\label{eq:estpwk2}
\norm{P_\W U^{*}e_l}_\H^2 \le 2 \mu^2 \eta_\Delta^2|\Delta|.
\end{equation}
 Thus, since $\|  g\|_\H = 1$, by \eqref{eq:withB2} we have
\[
|\langle D^{-*} P_{\W}^\perp  U^{\ast} e_l e_l^{*} \U P_{\W} g, e_j \rangle|  = |\langle g,P_{\W} U^{-1} e_l\rangle| |\langle U P_{\W}^\perp D^{-1} e_j, e_l \rangle |
\le \mu^2 B_\Delta \eta_\Delta \sqrt{2|\Delta|}.
\]
We have obtained that for $j\in J$ and $l=1, \ldots, N$
\[
\left|X_l^j \right|  \leq  \max\{\theta^{-1}(1-\theta), 1\}\mu^2 B_\Delta  \eta_\Delta \sqrt{2|\Delta|}   \leq  \theta^{-1}\mu^2 B_\Delta \eta_\Delta \sqrt{2|\Delta|}\k_1=: B,
\]
where in the last inequality we used the fact that $1\leq \k_1$. 

Now let $\Gamma \subseteq {J}$ be a set such that 
\begin{equation*}
\P \left(\sup_{j \in \Gamma} \left| \sum_{l= 1}^N X_l^j \right| \geq \frac t 2 \right) = 0\qquad\text{and}\qquad |\Gamma^c| \leq \tilde M\Bigl(\frac {t\theta}{2}\Bigr),
\end{equation*}
where $\Gamma^c:=J\setminus\Gamma$. By the Bernstein inequality \eqref{matrix} with $d=1$ we have
\begin{align*}
\P \left(\sup_{j \in J} \left| \sum_{l= 1}^N X_l^j \right|\geq \frac t 2 \right) &\leq \P \left(\sup_{j \in \Gamma^c} \left| \sum_{l= 1}^N X_l^j \right|\geq \frac t 2 \right)\\
&\leq 2\tilde M(\tfrac {t\theta}{2}) \exp \left(-\frac{t^2 \theta}{8\k_1\mu^2 B_\Delta( B_\Delta +\eta_\Delta\sqrt{2|\Delta|}t/6)} \right),
\end{align*}
which is the final estimate.

To finish the proof, we need to show that such $\Gamma$ exists. Note that, because of \eqref{eq:boundkU} and $\|g\|_\H = 1$ we have
\begin{align*}
\left| \sum_{l= 1}^N X_l^j \right| &= \left|\sum_{l= 1}^N\langle \theta^{-1} P_\Delta^\perp D^{-*} P^\perp_{\W} U^* (\delta_l-\theta)e_l e_l^* \U P_{\W} g,e_j\rangle \right|\\
&=\theta^{-1}\left| \langle g,P_{\W}U^{-1} \left(\textstyle\sum_{l}(\delta_l-\theta)e_l e_l^*\right) U P_{\W}^\perp D^{-1} P_\Delta^\perp e_j\rangle\right|\\
&\le\theta^{-1}\| P_{\W}U^{-1} \left(\textstyle\sum_{l}(\delta_l-\theta)e_l e_l^*\right) U P_{\W}^\perp D^{-1}  e_j\|_\H\\
&\le \theta^{-1}\sqrt{\k_1}\|  \left(\textstyle\sum_{l}(\delta_l-\theta)e_l e_l^*\right) U P_{\W}^\perp D^{-1}  e_j\|_2\\
&\leq \theta^{-1}\sqrt{\k_1}\| P_N U P^\perp_{\W} D^{-1} e_j\|_2\\
&\leq \theta^{-1}\sqrt{\k_1}\left(\| P_N U D^{-1} e_j\|_2 + \sqrt{\k_1}\| P_{{\widetilde{\W}}} D^{-1} e_j\|_{\H}\right).
\end{align*}
We then define
\begin{equation*}
\Gamma = \left\{ j \in {J} : \theta^{-1}\sqrt{\k_1}\left(\| P_N U D^{-1} e_j\|_2 + \sqrt{\k_1}\| P_{{\widetilde{\W}}} D^{-1} e_j\|_{\H}\right) < \frac{t}{2} \right\},
\end{equation*}
which is a finite set and satisfies $|\Gamma^c| \leq \tilde M(t\theta/2)$ by \eqref{eq:tildeM}. The proof follows.
\end{proof}

\begin{remark*}
Observe that in the above proof we  used the full generality of Definition~\ref{def:mu}: all four terms appear in the derivation.
\end{remark*}

\begin{proposition}\label{Proposition9.2}
For $g \in \W$ and $\left(4\sqrt{\sqrt{\k_2}\log(|\Delta|\k_1^2\k_2)}\right)^{-1}\le t\le 2\k_1$  we have
\begin{equation*}
 \P \left( \left\Vert \bigl(\theta^{-1} P_{\W} E_{\Omega} P_{\W}- P_ {\W} \bigr) g\right\Vert_\H > t \Vert g\Vert_\H \right)
 \le \exp \left(\frac{-t^2 \theta}{64|\Delta|\mu^2\eta_\Delta^2\k_1} + \frac{1}{4} \right).
\end{equation*}
\end{proposition}

\begin{proof}
Without loss of generality we assume that $\Vert g \Vert_\H=1$. For $l\in {L}$, let 
\[
\xi_l= P_{\W} U^{*}e_l, \qquad \alpha_l=P_{\W} U^{-1}e_l.
\]
We first make the following observations which will be useful along the proof, and follow from  \eqref{eq:estpwk2}, \eqref{eq:estpwk} and \eqref{eq:boundkU}:
\begin{align}
&\Vert \xk \Vert_\H^2 = \|P_{\W} U^{\ast} e_l\|_\H^2 \leq 2\eta_\Delta^2\mu^2|\Delta|\label{xk},\\
&\Vert \alpha_l \Vert_\H^2 = \|P_{\W} U^{-1} e_l\|_\H^2 \leq 2\eta_\Delta^2 \mu^2|\Delta|\label{ak},\\ 
&\sum_{l=1}^N \left|\langle\alpha_l, g \rangle \right|^2 = \sum_{l=1}^N \left|\langle e_l,U^{-*} P_{\W} g \rangle \right|^2 
 \leq \Vert U^{-*} P_{\W} g\Vert_2^2 \leq \Vert U^{-*} \Vert^2 \leq \k_1. \label{ek}
\end{align}

For $u,v\in\W$, let $u \otimes \overline v$ denote the continuous operator $\W\to\W$ defined by $(u \otimes \overline v)(w)=\langle w,v\rangle u$ for $w\in\W$ (note that $u\otimes \overline v$ is linear in $u$ and antilinear in $v$).
We have that
\begin{align*}
& \sum_{l=1}^N \xi_l \otimes \overline \alpha_l = P_{\W} U^* P_N U^{-*} P_{\W},  
\qquad {\sum_{l\in L\setminus\{1,\dots,N\}} \xk \otimes \overline \alpha_l = P_{\W} U^{*}P_N^\perp U^{-*} P_{\W}}, \\ 
& \theta^{-1} \sum_{l=1}^N \delta_l (\xi_l \otimes \overline\alpha_l) = \theta^{-1} P_{\W} U^* P_\Omega U^{-*} P_{\W}, 
\qquad P_{\W} = \sum\limits_{l {\in L}} \xk \otimes \overline \alpha_l.
\end{align*}
Hence, we have
\[
\begin{split}
& \left\Vert \left(\theta^{-1}P_{\W} U^{*} P_\Omega\U P_{\W} - P_{\W} \right)g \right\Vert_\H \\ &= \left\Vert \left(\sum_{l=1}^N (\theta^{-1}\delta_l -1)(\xk \otimes \overline \alpha_l) \right) g -\sum_{{l\in L\setminus\{1,\dots,N\}}} (\xk \otimes \overline \alpha_l) g \right\Vert_\H \\
& \leq \left\Vert \left(\sum_{l=1}^N (\theta^{-1}\delta_l -1)(\xk \otimes \overline \alpha_l) \right) g\right\Vert_\H  \!\!\!+ \left\Vert\left(P_{\W} U^{*}P_N^\perp\U P_{\W}\right)g \right\Vert_\H.
\end{split}
\]
Therefore, by the balancing property \eqref{bp1} it follows that
\[
\begin{split}
& \P \left(\left\Vert \left(\theta^{-1}P_{\W} U^{*}P_\Omega\U P_{\W} - P_{\W} \right)g \right\Vert_\H > t \right)\\
& \le \P \left(\left\Vert \left(\theta^{-1}P_{\W} U^{*}P_\Omega\U P_{\W} - P_{\W} \right)g \right\Vert_\H > \frac{t}{2}+  \left\Vert P_{\W} U^{*}P_N^\perp \U P_{\W}\right\Vert  \right)\\
& \leq \P \left(\left\Vert \sum_{l=1}^N (\theta^{-1}\delta_l -1)(\xk \otimes \overline \alpha_l) g \right\Vert_\H > \frac{t}{2} \right)
\end{split}
\]
for $t \geq \left(4\sqrt{\sqrt{\k_2}\log(|\Delta|\k_1^2\k_2)}\right)^{-1}$.

Let us estimate the above probability by using Lemma~\ref{lem:vectorbern}. We define
\[
X_l= (\theta^{-1}\delta_l -1)(\xk \otimes \overline \alpha_l)g\in\W\cong\mathbb{C}^d,
\]
with $d=\dim\W$.
First note that $\E(X_l) =0$. Next, observe that
\[
\Vert X_l \Vert_{\H}^2  = (\theta^{-1}\delta_l -1)^2 |\langle g,\alpha_l \rangle|^2 \norm{\xk}_\H^2  \le (\theta^{-1}\delta_l -1)^2   \Vert \alpha_l \Vert_\H^2 \norm{\xk}_\H^2.
\]
Thus, by \eqref{xk} and \eqref{ak} it follows that
\[
\Vert X_l \Vert_{\H}  \leq \max \{\theta^{-1}-1, 1\} \Vert \xk \Vert_\H \Vert \alpha_l \Vert_\H \leq 2\theta^{-1} |\Delta|\mu^2 \eta_\Delta^2 = :B.
\]
In addition, since $\E  (\theta^{-1}\delta_l -1)^2 = \theta^{-1}-1$, by \eqref{xk} and \eqref{ek} we obtain
\[
\sum_{l=1}^N \ \E\left(\Vert X_l \Vert^2_\H \right)   \leq 2(\theta^{-1}-1)|\Delta|\mu^2\eta_\Delta^2\sum_{l=1}^N |\langle \alpha_l, g \rangle|^2 \leq2 \theta^{-1}|\Delta|\mu^2\eta_\Delta^2\k_1=:\sigma^2.
\]
Therefore, applying the Vector Bernstein inequality \eqref{vector}  we get the desired estimate.
\end{proof}

The next proposition involves an operator containing $U^{-1}P_\Omega U^{-*}$.

\begin{proposition}\label{prop:new4}
For $g \in \W$, we have
\begin{equation*}
\P \left( \left\| \theta^{-1}P_\W U^{-1} P_\Omega U^{-*}P_\W g \right\|_\H > 2\k_1 \|g\|_\H\right) \leq \exp\left( \frac{- \k_1 \theta}{16 \mu^2 \eta_\Delta^2 |\Delta|}+\frac 1 4\right).
\end{equation*}
\end{proposition}
\begin{proof}
Without loss of generality we assume that $\Vert g \Vert_\H=1$. For $l\in {L}$, let $\alpha_l=P_{\W} U^{-1}e_l$. Arguing as in the proof of Proposition \ref{Proposition9.2} we have
\begin{align*}
&\|\theta^{-1}P_\W U^{-1} P_\Omega U^{-*}P_\W g\|_\H  \\&\leq \left\| \left(\sum_{l=1}^N (\theta^{-1}\delta_l -1)(\alpha_l \otimes \overline \alpha_l) \right) g \right\|_\H+\|P_\W U^{-1}P_N U^{-*}P_\W g\|_\H\\
&\leq \left\| \left(\sum_{l=1}^N (\theta^{-1}\delta_l -1)(\alpha_l \otimes \overline \alpha_l) \right) g \right\|_\H+\k_1.
\end{align*}
Therefore it follows that
\[\P \left( \left\| \theta^{-1}P_\W U^{-1} P_\Omega U^{-*}P_\W g \right\|_\H > 2\k_1\right)
\leq  \P \left( \left\|\sum_{l=1}^N (\theta^{-1}\delta_l -1)(\alpha_l \otimes \overline \alpha_l)  g \right\|_\H > \k_1\right).
\]
We will bound this probability using Lemma \ref{lem:vectorbern}. We define
\[
X_l = (\theta^{-1}\delta_l -1)(\alpha_l \otimes \overline \alpha_l)  g \in \W \cong\mathbb{C}^d,
\]
with $d=\dim\W$.
First note that $\E(X_l) =0$. Next, observe that
\[
\Vert X_l \Vert_{\H}^2  = (\theta^{-1}\delta_l -1)^2 |\langle g,\alpha_l \rangle|^2 \norm{\alpha_l}_\H^2  \le (\theta^{-1}\delta_l -1)^2   \Vert \alpha_l \Vert_\H^4.
\]
Thus, by \eqref{ak} it follows that
\[
\Vert X_l \Vert_{\H}  \leq \max \{\theta^{-1}-1, 1\} \Vert \alpha_l \Vert_\H^2  \leq 2\theta^{-1} |\Delta|\mu^2 \eta_\Delta^2 = :B.
\]
In addition, since $\E  (\theta^{-1}\delta_l -1)^2 = \theta^{-1}-1$, by \eqref{ek} we obtain
\[
\sum_{l=1}^N \ \E\left(\Vert X_l \Vert^2_\H \right)   \leq 2(\theta^{-1}-1)|\Delta|\mu^2\eta_\Delta^2\sum_{l=1}^N |\langle \alpha_l, g \rangle|^2 \leq2 \theta^{-1}|\Delta|\mu^2\eta_\Delta^2\k_1=:\sigma^2.
\]
Therefore, applying the Vector Bernstein inequality \eqref{vector} for $t=\k_1$ we get the desired estimate.
\end{proof}

In the next result we  deal with a matrix operator containing $U^{-1}P_\Omega U$.

\begin{proposition}\label{Theorem9.3}
We have
\begin{equation*}
 \P \left( \left\Vert \bigl(\theta^{-1} P_{\W} U^{-1}P_\Omega U P_{\W}- P_{\W} \bigr) \right\Vert_{\H\to \H} >  \frac 1 2 \right)  \leq {4}|\Delta| \exp \left(\frac{-3 \theta}{208|\Delta|\mu^2 \eta_\Delta^2\k_1}  \right).
\end{equation*}
\end{proposition}

\begin{proof}
We consider $\xi_l= P_{\W} U^{*}e_l$, $\alpha_l=P_{\W} U^{-1}e_l$, and arguing  as in Proposition~\ref{Proposition9.2}, we arrive to
\begin{align*}
& \P \left( \left\Vert \bigl(\theta^{-1} P_{\W} U^{-1}P_\Omega U P_{\W}- P_{\W} \bigr) \right\Vert >  \frac 1 2 \right) \\
& \le \P \left(\left\Vert \theta^{-1}P_{\W} U^{-1}P_\Omega U P_{\W} - P_{\W} \right\Vert > \frac 1 4 +  \left\Vert P_{\W} U^{-1}P_N^\perp U P_{\W} \right\Vert  \right)\\
& \leq \P\left(  \left\Vert \sum_{l=1}^N (\theta^{-1}\delta_l -1)(\alpha_l \otimes \overline \xi_l)  \right\Vert > \frac 1 4 \right),
\end{align*}
where the last probability of the above inequality will be estimated by using Lemma~\ref{lem:matrixbern}.

Let us define now
\[
X_l = (\theta^{-1}\delta_l -1)(\alpha_l \otimes \overline \xi_l)\colon \W\to\W.
\]
Since $\W$ is finite dimensional, $X_l$ may be identified with an element in $\mathbb{C}^{d\times d}$, where $d=\dim\W\le {2} |\Delta|$.
We have $\E (X_l)=0$. Further, since $1\leq \k_1$ and by \eqref{xk} and \eqref{ak}, it follows that
\[
\Vert X_l \Vert  \leq \max\{\theta^{-1}-1, 1\} \lVert \alpha_l \rVert_\H \lVert \xk \rVert_\H \leq 2 \theta^{-1}\mu^2\eta_\Delta^2 |\Delta|\k_1 = : B.
\]
We next study $\E(X_l^{*}X_l)$ and $\E(X_l X_l^{*})$. Since $X_l^{*}= (\theta^{-1}\delta_l -1)\xk \otimes \overline \alpha_l$, we have
\begin{align*}
 X_l^{*} X_l & = (\theta^{-1}\delta_l -1)^2 (\xk \otimes \overline \alpha_l) (\alpha_l \otimes \overline \xi_l) = (\theta^{-1}\delta_l -1)^2 \Vert \alpha_l \Vert_\H^2\, \xk \otimes \overline \xi_l,\\
 X_l X_l^{*} & = (\theta^{-1}\delta_l -1)^2 (\alpha_l \otimes \overline \xi_l) (\xk \otimes \overline \alpha_l) = (\theta^{-1}\delta_l -1)^2 \Vert \xk \Vert_\H^2\, \alpha_l  \otimes \overline \alpha_l.
 \end{align*}
As a consequence, since $\xi_l = P_{\W} U^{*}e_l$, for every $h\in\H$ with $\norm{h}_\H=1$ we have
 \[
 \sum_{l=1}^N \E(X^*_l X_l) h 
 = (\theta^{-1}-1)P_{\W} U^{*}\sum_{l=1}^N \Vert \alpha_l \Vert_\H^2 \langle  h,P_{\W} U^{*} e_l \rangle_\H e_l,
 \]
 and, by \eqref{eq:boundkU} and \eqref{ak}, we readily deduce
 \[
 \norm{\sum_{l=1}^N \E(X^*_l X_l) h }_\H\le 2\theta^{-1} \sqrt{\k_1} \eta_\Delta^2 \mu^2|\Delta|\left( \sum_{l=1}^N |\langle  U P_{\W} h,e_l \rangle_\H|^2\right)^{\frac12}
 \le  2\theta^{-1} \k_1 \eta_\Delta^2 \mu^2|\Delta|.
 \]
 Arguing in the same way, we obtain
 \[
 \norm{\sum_{l=1}^N \E(X_l X_l^*) }
 \le  2\theta^{-1} \k_1 \eta_\Delta^2 \mu^2|\Delta|.
 \]

Hence we can choose
\[
\sigma^2 := 2\theta^{-1}\mu^2 \eta_\Delta^2|\Delta|\k_1,
\]
and applying the Bernstein inequality \eqref{matrix} for $t = \frac 1 4$,  we deduce the result.
\end{proof}

In the next result we deal with a matrix operator containing $U^{-1}P_\Omega U^{-*}$.

\begin{proposition}\label{prop:new}
We have
\begin{equation*}
 \P \left( \left\Vert\theta^{-1} P_{\W} U^{-1}P_\Omega U^{-*} P_{\W} \right\Vert_{\H\to \H} >  2\k_1 \right)  \leq {4}|\Delta| \exp \left(\frac{-3\k_1\theta }{16|\Delta|\mu^2 \eta_\Delta^2}  \right).
\end{equation*}
\end{proposition}

\begin{proof}
The proof is very similar to that of Proposition~\ref{Theorem9.3}, and only a sketch will be provided. 

Set $\alpha_l=P_{\W} U^{-1}e_l$.  The bound $\left\Vert P_{\W} U^{-1}P_N U^{-*} P_{\W} \right\Vert\le\k_1$ yields 
\begin{equation*}
\begin{split}
 \P  \left( \left\Vert \theta^{-1} P_{\W} U^{-1}P_\Omega U^{-*} P_{\W}\right\Vert>  2\k_1 \right)  &\le \P \left(\left\Vert  P_{\W} U^{-1}(\theta^{-1}P_\Omega-P_N) U^{-*} P_{\W} \right\Vert > \k_1 \right)\\
& = \P\left(  \left\Vert \sum_{l=1}^N X_l  \right\Vert >\k_1\right),
\end{split}
\end{equation*}
where 
\[
X_l = (\theta^{-1}\delta_l -1)(\alpha_l \otimes \overline \alpha_l)\colon \W\to\W.
\]
Since $\W$ is finite dimensional, $X_l$ may be identified with an element in $\mathbb{C}^{d\times d}$, where $d=\dim\W\le {2} |\Delta|$.
We have $\E (X_l)=0$. Further, by  \eqref{ak}, it follows that
\[
\Vert X_l \Vert  \leq \max\{\theta^{-1}-1, 1\} \lVert \alpha_l \rVert_\H^2  \leq 2 \theta^{-1}\mu^2\eta_\Delta^2 |\Delta|= : B.
\]
Note that $X_l^*=X_l$. We next study $\E(X_l^{*}X_l)=\E(X_l X_l^{*})$. We have
$
 X_l^{*} X_l  = (\theta^{-1}\delta_l -1)^2 \Vert \alpha_l \Vert_\H^2\, \alpha_l \otimes \overline \alpha_l$. 
As a consequence,  for every $h\in\H$ with $\norm{h}_\H=1$ we have
 \[
 \sum_{l=1}^N \E(X^*_l X_l) h 
 = (\theta^{-1}-1)P_{\W} U^{-1}\sum_{l=1}^N \Vert \alpha_l \Vert_\H^2 \langle  h,P_{\W} U^{-1} e_l \rangle_\H e_l,
 \]
 and, by \eqref{eq:boundkU} and \eqref{ak}, we readily deduce
 \[
 \begin{split}
 \norm{\sum_{l=1}^N \E(X^*_l X_l) h }_\H&\le 2\theta^{-1} \sqrt{\k_1} \eta_\Delta^2 \mu^2|\Delta|\left( \sum_{l=1}^N |\langle  U^{-*} P_{\W} h,e_l \rangle_\H|^2\right)^{\frac12}
 \\&\le  2\theta^{-1} \k_1 \eta_\Delta^2 \mu^2|\Delta|.
 \end{split}
 \]
 
Hence we can choose
$
\sigma^2 := 2\theta^{-1}\mu^2 \eta_\Delta^2|\Delta|\k_1,
$
and applying the Bernstein inequality \eqref{matrix} for $t=\k_1$,  we deduce the result.
\end{proof}

We conclude this subsection with the following estimate.

\begin{proposition}\label{prop:unifoff}
 We have
\begin{multline*}
 \P\left(\sup_{j \in \Delta^c } \Vert \theta^{-1} P_{\{j\}} D^{-*} P_{\W}^\perp U^{*}P_{\Omega}U P_{\W}^\perp D^{-1} P_{\{j\}} \Vert_{\ell^2(J)\to\ell^2(J)} >  2\k_1\k_2  \right)
\\ \leq 2 \tilde M(\theta) \exp \left(-\frac{ 3\theta \k_1\k_2}{8 \mu^2 B_\Delta^2}  \right).
\end{multline*}
\end{proposition}
\begin{proof}
Fix $j \in \Delta^c$ and let $T = U P_\W^\perp D^{-1}$. We have 
\begin{align*}
\|\theta^{-1} P_{\{j\}} T^* P_{\Omega}T P_{\{j\}}\| \leq \left|\sum_{l=1}^N Y^j_l\right| +  \|P_{\{j\}} T^* P_{N}T P_{\{j\}}\|,
\end{align*}
where 
\begin{align*}
Y_l^j =(\theta^{-1}\delta_l - 1) \langle T^* (e_l \otimes \overline{e_l}) T e_j,e_j\rangle
= (\theta^{-1}\delta_l - 1) |\langle e_l,Te_j\rangle|^2.
\end{align*}
Note that $\E (Y_l)=0$. Since, from our main assumptions,
\begin{align*}
\|P_{\{j\}} T^* P_{N}T P_{\{j\}}\| \leq \k_1 \k_2,
\end{align*}
we obtain
\[
\P  \left(\Vert \theta^{-1}P_{\{j\}} T^* P_{\Omega}T P_{\{j\}} \Vert > 2 \k_1\k_2  \right)  
\leq \P\left(  \left| \sum_{l=1}^N Y^j_l \right| >  \k_1\k_2 \right).
\]
Next, by \eqref{eq:withB2} we have
\[
|Y^j_l| = | (\theta^{-1}\delta_l - 1)| |\langle e_l,Te_j\rangle|^2 
=| (\theta^{-1}\delta_l - 1)| |\langle e_l, U P_\W^\perp D^{-1} e_j\rangle|^2 \leq \theta^{-1}\mu^2 B_\Delta^2= : B.
\]
In addition, using again \eqref{eq:withB2} yields
\begin{align*}
\sum_{l=1}^{N} \E (|Y^j_l|^2) & = (\theta^{-1}-1) \sum_{l = 1}^N |\langle e_l,Te_j\rangle|^4  \\
&\leq (\theta^{-1}-1) \mu^2 B_\Delta^2 \sum_{l = 1}^N |\langle e_l,Te_j\rangle|^2\\
& \leq \theta^{-1} \mu^2 B_\Delta^2 \|Te_j\|_2^2 \\
& \leq \theta^{-1} \mu^2  B_\Delta^2 \k_1 \k_2 = : \sigma^2.
\end{align*}

Now, assume that there exists a non-empty set $\Gamma \subseteq {J}$ such that
\begin{equation*} \label{est:inf2}
\P \left( \sup_{j \in \Gamma}\Vert \theta^{-1} P_{\{j\}} T^{*}P_{\Omega}T P_{\{j\}} \Vert >  2\k_1\k_2 \right) = 0\qquad\text{and}\qquad |\Gamma^c| \leq \tilde M(\theta). 
\end{equation*}
By the matrix Bernstein inequality \eqref{matrix} for $d =1$ and $t =\k_1\k_2$  and the union bound, we obtain
\[
\begin{split}
\P &\left( \sup_{j \in \Delta^c} \Vert \theta^{-1} P_{\{j\}} T^{*}P_{\Omega}T P_{\{j\}} \Vert >  2\k_1\k_2 \right) \\ &= \P \left( \sup_{j \in \Delta^c \cap \Gamma^c}\Vert \theta^{-1} P_{\{j\}} T^{*}P_{\Omega}T P_{\{j\}} \Vert >  2\k_1\k_2 \right)\\
& \leq 2 \tilde M(\theta) \exp \left(-\frac{ 3\theta \k_1\k_2}{8 \mu^2 B_\Delta^2}  \right),
\end{split}
\]
which is our final estimate.

We only have to show the existence of $\Gamma$ and provide a bound on $|\Gamma^c|$.
Note that
\[
\begin{split}
\Vert \theta^{-1}P_{\{j\}} T^* P_{\Omega}T P_{\{j\}} \Vert &\leq \theta^{-1}\sqrt{\k_1 \k_2}\Vert P_{\Omega}T P_{\{j\}} \Vert\\
&\leq \theta^{-1}\sqrt{\k_1 \k_2}\Vert P_{N}T e_j \Vert_2\\
& = \theta^{-1}\sqrt{\k_1 \k_2}\Vert P_{N} U P_\W^\perp D^{-1}e_j \Vert_2\\
&\leq \theta^{-1}\sqrt{\k_1 \k_2} \left(\| P_N U D^{-1} e_j\|_2 + \sqrt{\k_1}\| P_{{\widetilde{\W}}} D^{-1} e_j\|_{\H}\right).
\end{split}
\]
We then define, similarly to Proposition~\ref{prop9.1},
\begin{equation*}
\Gamma = \left\{ j \in {J} : \theta^{-1}\sqrt{\k_1}\left(\| P_N U D^{-1} e_j\|_2 + \sqrt{\k_1}\| P_{{\widetilde{\W}}} D^{-1} e_j\|_{\H}\right) < 2 \k_1\sqrt{\k_2} \right\},
\end{equation*}
which satisfies $|\Gamma^c| \leq \tilde M(2 \k_1 \sqrt{\k_2}\theta) \leq \tilde M(\theta)$ by \eqref{eq:tildeM}. The proof follows.
\end{proof}

\subsection{The dual certificate}\label{sub:dual}

We now show how the existence of a \emph{dual certificate} $\rho$ satisfying certain properties guarantees exact recovery up to measurement noise.
The result is standard, and we follow closely \cite[Proposition~6.1]{Poon2015}. An analogous result holds true for more general operators \cite{grasmair2008,grasmair-etal-2011,grasmair-etal-2011b,haltmeier2013}, and this would allow considering the case where the inverse of $U$ is unbounded: this is left for future work.

\begin{proposition}\label{prop:dual}
Assume that Hypothesis~\ref{hp1} holds true, and let $U$ and $D$ denote the corresponding analysis operators, satisfying the bounds given in \eqref{eq:boundk}. Let $\Delta \subseteq J$ and $\Omega$ be a finite subset of $L$. Let $g_0 \in \H$ and $\eta \in\ell^2({L})$ be such that $\norm{\eta}_2\le\epsilon$ for some $\epsilon\ge 0$. Let $\zeta = P_{\Omega} U g_0 + \eta$ be the known noisy measurement. 
Assume that there exist $\rho = U^* P_\Omega \rho'$ for some $\rho' \in \ell^2({L})$, $\L>0$ and $0 <\theta \leq 1$ with the following properties:
\begin{enumerate}[(i)]

\item\label{eq:dual1} $\Vert \left( \theta^{-1}P_{\W}U^{-1}P_{\Omega}UP_{\W}\right)^{-1} \Vert_{\W\to\W} \leq 2$,
\item\label{eq:dual2}  $\|\theta^{-1}P_\W U^{-1} P_\Omega U^{-*} P_\W \|_{\W\to\W}\le 2\k_1$,
\item\label{eq:dual3} $\sup_{j \in \Delta^c} \Vert \theta^{-1}P_{\{j\}} D^{-*} P_\W^\perp U^{*}P_{\Omega}U P_{\W}^{\perp} D^{-1}P_{\{j\}} \Vert_{\ell^2(J)\to\ell^2(J)} \leq 2 \k_1 \k_2$,
\item\label{eq:dual4} $\Vert P_{\W}\rho - D^{\ast} \sgn(P_{\Delta} D g_0)\Vert_\H \leq \frac{1}{16\k_1\sqrt{\k_2}}$,
\item\label{eq:dual5} $\Vert P_{\Delta}^{\perp} D^{-*} P_{\W}^{\perp} \rho \Vert_{l^\infty{(J)}} \leq \frac{1}{4}$,
\item\label{eq:dual6} $\| \rho'\|_{{\ell^2(L)}} \leq \L \sqrt{\k_1 \k_2|\Delta|}$.

\end{enumerate}

Let $g \in \H$ be a minimizer of the problem
\begin{align*}
\inf_{ \substack{
 g \in \H \\  D  g \in \ell^1({J})
}  } \Vert D  g \Vert_1 \quad \text{subject to } \norm{P_\Omega U  g - \zeta}_2\le\epsilon.
\end{align*}
Then
\[
\Vert g - g_0 \Vert_\H \leq 20\k_1\sqrt{\k_2} \|P_{\Delta}^\perp D g_0\|_1+  \epsilon \sqrt{\k_1} \left(10\,\theta^{-1/2}+20Q\k_1\k_2\sqrt{|\Delta|} \right).
\]
\end{proposition}

\begin{proof}
We start from the following identity, for any $\tilde g \in \H$:
\begin{equation} \label{eq:idpw}
P_{\W}^\perp \tilde g =P_{\W}^\perp D^{-1} D \tilde g =P_{\W}^\perp D^{-1} P_{\Delta}^\perp D \tilde g.
\end{equation}
This follows from the fact that $D^{-1} D$ is the identity and that $P_{\W}$ is the orthogonal projection onto $R(D^{-1} P_{\Delta}) + R(D^{*} P_{\Delta})$. From estimates \eqref{eq:boundkD} we obtain, for any $\tilde g \in \H$,
\begin{equation}\label{eq:decomposition}
\| \tilde g \|_\H \leq   \| P_{\W} \tilde g\|_\H + \| P_{\W}^\perp \tilde g \|_\H\leq \| P_{\W} \tilde g \|_\H +\sqrt{\k_2} \| P_{\Delta}^\perp D \tilde g\|_1.
\end{equation}
From the last inequality applied to $h = g-g_0$, we see that it is enough to bound $\| P_{\W} h \|_\H$ and $\| P_{\Delta}^\perp D h\|_1$ in order to finish the proof.
Let us start from $\| P_{\W} h \|_\H$.

First note that since $\norm{P_\Omega U  g - \zeta}_2\le\epsilon$  we have 
\begin{equation}\label{eq:2eps}
\|P_{\Omega} U h \|_2 \leq 2 \epsilon.
\end{equation}
Set $T=\theta^{-1/2}P_\W U^{-1}P_\Omega $. By \ref{eq:dual2} we have $\norm{T}=\sqrt{\norm{TT^* }}\le \sqrt{2\k_1}$. Thus, using \eqref{eq:2eps}
and \ref{eq:dual1} we find
\begin{align*}
\| P_{\W} h \|_\H &= \| (P_{\W}U^{-1}P_{\Omega} U P_{\W})^{-1}P_{\W}U^{-1}P_{\Omega} U P_{\W} h\|_\H \\
&\leq 2 \theta^{-1/2} \| \theta^{-1/2}P_{\W}U^{-1}P_{\Omega} U (h- P_{\W}^\perp h)\|_\H \\
&\leq 2\theta^{-1/2}\sqrt{2\k_1}( 2 \epsilon  + \|P_{\Omega} U P_{\W}^\perp h\|_2).
\end{align*}
We bound the last term as follows. Set $T_j=\theta^{-1/2} P_{\Omega} U P_{\W}^\perp D^{-1}P_{\{j\}}$. By \ref{eq:dual3} we have $\norm{T_j}=\sqrt{\norm{T_j^* T_j}}\le \sqrt{2\k_1\k_2}$. Then
\[
\begin{split}
\|\theta^{-1/2}P_{\Omega} U P_{\W}^\perp h\|_2 &= \|\theta^{-1/2} P_{\Omega} U P_{\W}^\perp D^{-1} P_{\Delta}^\perp D h\|_2\\
& \leq \sup_{j \in \Delta^c} \| \theta^{-1/2} P_{\Omega} U P_{\W}^\perp D^{-1}e_j \|_2 \| P_{\Delta}^\perp D h\|_1\\
& = \sup_{j \in \Delta^c} \| T_j \| \| P_{\Delta}^\perp D h\|_1\\
& \leq \sqrt{2\k_1\k_2}\| P_{\Delta}^\perp D h\|_1,
\end{split}
\]
where we used identity \eqref{eq:idpw}. Thus we have found
\begin{equation} \label{eq:estpwh}
\| P_{\W} h \|_\H \leq 2\sqrt{2\k_1}( 2 \epsilon  \theta^{-1/2}+ \sqrt{2\k_1\k_2}\| P_{\Delta}^\perp D h\|_1).
\end{equation}

We now pass to the estimate of $\| P_{\Delta}^\perp D h\|_1$. Note that
\begin{align*}
\| D g\|_1 &= \|P_{\Delta}^\perp D (g_0+h)\|_1+\|P_{\Delta} D (g_0+h)\|_1 \\
&\geq \|P_{\Delta}^\perp D h \|_1-\|P_{\Delta}^\perp D g_0\|_{1}+\|P_{\Delta} D g_0\|_1 +\Ree \langle P_{\Delta}D h, \sgn(P_{\Delta} D g_0)\rangle\\
&\ge  \|P_{\Delta}^\perp D h \|_1-2 \|P_{\Delta}^\perp D g_0\|_1+\| D g_0\|_1 -| \langle P_{\Delta}D h, \sgn(P_{\Delta} D g_0)\rangle|.
\end{align*}
Singe $g$ is a minimizer, we find that
\begin{equation*}
\| P_{\Delta}^\perp D h\|_1 \leq 2 \|P_{\Delta}^\perp D g_0\|_1 + |\langle P_{\Delta}D h, \sgn(P_{\Delta} D g_0)\rangle|.
\end{equation*}
Now, since $\rho = U^{\ast} P_{\Omega} \rho'$ and by \eqref{eq:idpw}, \eqref{eq:2eps}, \eqref{eq:estpwh}, \ref{eq:dual4}, \ref{eq:dual5} and \ref{eq:dual6} we have
\begin{align*}
&|\langle P_{\Delta}D h, \sgn(P_{\Delta} D g_0)\rangle| = |\langle h, D^{\ast}\sgn(P_{\Delta} D g_0)\rangle|\\
&\qquad \leq |\langle h , D^{\ast} \sgn(P_{\Delta} D g_0)- P_{\W} \rho\rangle|+|\langle h, \rho \rangle| + |\langle h ,P_{\W}^\perp \rho\rangle|\\
&\qquad \leq \| P_{\W} h\|_\H \|D^{\ast}\sgn(P_{\Delta} D g_0)-P_{\W} \rho\|_\H + \|P_{\Omega} U h\|_2 \| \rho'\|_2+ |\langle D^{-1} P_{\Delta}^\perp D h, P_{\W}^\perp \rho \rangle |\\
&\qquad \leq \frac{1}{16\k_1\sqrt{\k_2}}\| P_{\W} h\|_\H + 2\epsilon \L \sqrt{\k_1 \k_2 |\Delta|} + \frac 1 4 \| P_{\Delta}^\perp D h\|_1\\
&\qquad \le \frac{1}{8\k_1\sqrt{\k_2}} \sqrt{2\k_1}( 2 \epsilon  \theta^{-1/2}+ \sqrt{2\k_1\k_2}\| P_{\Delta}^\perp D h\|_1) + 2\epsilon \L \sqrt{\k_1 \k_2 |\Delta|} + \frac 1 4 \| P_{\Delta}^\perp D h\|_1 \\
&\qquad \leq \epsilon  \left(\frac{\theta^{-1/2}}{2\sqrt{2\k_1\k_2}} + 2\L\sqrt{\k_1\k_2|\Delta|} \right)+\frac 1 2 \| P^\perp_{\Delta} D h\|_1.
\end{align*}
Thus we have obtained
\begin{equation*}
\| P_{\Delta}^\perp D h\|_1 \leq 4\| P_{\Delta}^\perp D g_0\|_1 + \epsilon  \left(\frac{\theta^{-1/2}}{\sqrt{2\k_1\k_2}} + 4\L\sqrt{\k_1\k_2|\Delta|} \right),
\end{equation*}
which by \eqref{eq:decomposition} and \eqref{eq:estpwh} yields the final estimate
\[
\begin{split}
\|h\|_{\H} &\leq \| P_{\W} h\|_\H + \sqrt{\k_2}\| P_{\Delta}^\perp D h\|_1 \\
&\leq 4\sqrt{2\k_1} \epsilon  \theta^{-1/2}+ (4\k_1\sqrt{\k_2}+\sqrt{\k_2})\| P_{\Delta}^\perp D h\|_1\\
& \leq 4\sqrt{\k_2}(4\k_1+1) \|P_{\Delta}^\perp D g_0\|_1 \\ &\qquad\qquad+  \epsilon \sqrt{\k_1} \left(\theta^{-1/2}\left( 6\sqrt{2}+\frac{1}{\sqrt{2}\k_1}\right)+(4\k_1+1)4Q\k_2\sqrt{|\Delta|} \right).
\end{split}
\]
This concludes the proof.
\end{proof}

By using the results of $\S$\ref{sub:four},
we now show that the dual certificate $\rho$ can be constructed. The proof is based on a \emph{golfing scheme} \cite{G,2015-gross-krahmer-kueng}.

\begin{proposition}\label{prop:golfing}
Assume that Hypothesis~\ref{hp1} holds true, and let $U$ and $D$ denote the corresponding analysis operators, satisfying the bounds given in \eqref{eq:boundk}. Let $M\in {J}$, $\omega \geq 1$ and $\Delta \subseteq \{1, \ldots, M\}$ be such that $|\Delta|\ge 3$. Let $N$ satisfy the balancing property with respect to $U$, $D$,  $M$ and $|\Delta|$.  Let $\Omega \subseteq \{1,\dots,N\}$ be chosen uniformly at random with $|\Omega|=m$. Take $g_0 \in \H$. If
\[
m \geq C \mu^2 \eta_{\Delta}^2 |\Delta| \k_1\k_2 \omega^2 B_{\Delta}^2 N \log \left(\k_1\k_2 \tilde M\Bigl(\tfrac{C' m}{N\omega \sqrt{|\Delta|\k_2}}\Bigr)\right),
\] 
then, with probability exceeding $1-e^{-\omega}$, there exist $\rho = U^* P_\Omega \rho'$ for some $\rho' \in \ell^2({L})$ and $\L\le C'''\sqrt{\omega\frac{N}{m}}$
satisfying the properties \ref{eq:dual1}-\ref{eq:dual6} of Proposition~\ref{prop:dual}, with $\theta=m/N$, where $C,C',C'''>0$ are  universal constants. 
\end{proposition}

\begin{proof}
The proof is based on a recursive procedure to construct a sequence of vectors $\{Y_i\}$ converging to the dual certificate $\rho$  with high probability.

The set $\Omega \subseteq \{1, \ldots, N\}$ is chosen uniformly at random with $|\Omega| = m$. It is well known (see \cite[Section II.C]{CRT}) that we may, without loss of generality, replace this way of choosing $\Omega$ with the model $\{1,\ldots,N\} \supseteq \Omega \sim Ber(\theta)$ for $\theta = m/N$ ($\theta$ will have this value throughout the proof). This is equivalent to choosing $\Omega$ as 
$$
\Omega = \Omega_1 \cup \ldots \cup \Omega_{l'}$$
with $\Omega_{l'}$ following a Bernoulli distribution as explained below (see also \cite[Section 9.1]{AH}). The main difference with the golfing scheme in \cite{AH} is that the number $l'$ of sampled sets is greater than $l$, the number of iterations in our recursive scheme (both to be determined later). 
In fact, given $q_i$ for $i=1, \ldots, l$, we will sample $l' \geq l$ sets distributed as $Ber(q_i)$, 
and will keep only $l$ of them for the construction of the certificate.

To initialize the iterations, set 
$$
Y_0  = 0,
$$
and define
\begin{equation}\label{eq:Zi}
Z_i = D^{\ast}\sgn (P_{\Delta}D g_0) - P_{\W} Y_i, \qquad i=0,\dots,l.
\end{equation}
Let us define the sequence $\{Y_i\}_{i=1}^l$ iteratively as follows. Given $q_i$, for $j=1,2,\dots$ we choose $\Omega^j_i\subseteq \{1,\dots,N\}$ at random such that $\Omega^j_i\sim Ber(q_i)$. Let $E_{\Omega_i^j}= U^{*}P_{\Omega_i^j}\U$. We repeat the choice for $j=1,2,\dots$ until the conditions
 \begin{align}
 & \left\Vert \bigl( P_ {\W} - q_i^{-1} P_{\W} E_{\Omega^j_i} P_{\W}
  \bigr) Z_{i-1}\right\Vert_\H \leq \alpha_i \Vert Z_{i-1}\Vert_\H, \label{c1}\\
  & \left\Vert  q_{i}^{-1} P_\Delta^{\perp}D^{-*} P^\perp_{\W} E_{\Omega^j_i}  Z_{i-1}\right\Vert _\infty \leq \beta_i \Vert  Z_{i-1} \Vert_\H, \label{c2}\\
& \left\|q_i^{-1}P_\W U^{-1} P_{\Omega_i^j}U^{-*} P_\W Z_{i-1}\right\|_\H \leq 2\k_1 \|Z_{i-1}\|_\H, \label{c3} 
\end{align}
hold for some parameters $\alpha_i, \beta_i \in \R$ that will be chosen later. We set
\[
r_i = \min\{j=1,2,\ldots: \text{\eqref{c1}, \eqref{c2} and \eqref{c3} are satisfied}\},
\]
namely $r_i$ denotes the number of repetitions of the $i$-th step. We also set
\[
\Omega=\bigcup_{i=1}^l \bigcup_{j=1}^{r_i} \Omega_i^j, \qquad \Omega_i = \Omega_i^{r_i}, \qquad Y_i = \sum_{k=1}^i q_k^{-1} E_{\Omega_k}Z_{k-1},
\]
and
\begin{equation*}
\rho  = Y_l,\qquad \rho' = \sum_{i=1}^l q_{i}^{-1} P_{\Omega_i}U^{-*}Z_{i-1},
\end{equation*}
so that $\rho = U^{*}P_\Omega \rho'$.

The identities \eqref{eq:Zi} and $Y_i = Y_{i-1} + q_i^{-1}E_{\Omega_i}Z_{i-1}$ yield
\[
Z_i  
  = Z_{i-1} -  q_{i}^{-1}P_{\W} E_{\Omega_i}P_{\W} Z_{i-1}
  = (P_{\W}- q_{i}^{-1}P_{\W} E_{\Omega_i}P_{\W}) Z_{i-1}. 
\]
Thus by \eqref{eq:boundkD} and \eqref{c1} it follows that
\[
\Vert Z_i \Vert_\H \leq \alpha_i \Vert Z_{i-1} \Vert_\H \leq \prod_{j=1}^i \alpha_j \Vert Z_0 \Vert_\H \le \sqrt{\k_2|\Delta|} \prod_{j=1}^i \alpha_j,
\]
which together with $Z_l =D^\ast\sgn(P_\Delta D g_0)-P_{\W} \rho$  gives
\[
\Vert D^\ast\sgn(P_\Delta D g_0)-P_{\W} \rho \Vert_\H \leq \sqrt{ \k_2|\Delta|} \prod_{i=1}^l \alpha_i.
\]
Moreover by \eqref{c2} we have,
\[
\Vert  P_\Delta^{\perp} D^{-*} P_{\W}^\perp \rho \Vert_\infty  \leq \sum_{i=1}^l \left\Vert q_i ^{-1} P_\Delta^{\perp} D^{-*} P_{\W}^\perp E_{\Omega_i}Z_{i-1} \right\Vert_\infty
\leq \sqrt{|\Delta|\k_2} \sum_{i=1}^l \beta_i \prod_{j=1}^{i-1} \alpha_j.
\]
(For $i=1$ we set $\Pi_{j=1}^{i-1} \alpha_j=1$.) Now let $\rho'_i =  q_{i}^{-1} P_{\Omega_i} U^{-*}Z_{i-1}$ for $i =1,\dots,l$. By \eqref{c3} we have
\begin{align*}
\| \rho_i'\|^2_2 &\leq q_{i}^{-2} \langle P_{\Omega_i} U^{-*}Z_{i-1},P_{\Omega_i} U^{-*}Z_{i-1} \rangle\\
& = q_{i}^{-1} \langle q_{i}^{-1} U^{-1} P_{\Omega_i} U^{-*}Z_{i-1},Z_{i-1} \rangle\\
& =  q_{i}^{-1} \langle q_{i}^{-1} P_\W U^{-1} P_{\Omega_i} U^{-*}P_\W  Z_{i-1},Z_{i-1} \rangle\\
&\leq q_{i}^{-1} \|q_{i}^{-1} P_\W U^{-1} P_{\Omega_i} U^{-*}P_\W  Z_{i-1}\|_\H \|Z_{i-1}\|_\H \\
&\leq 2\k_1q_{i}^{-1}\|Z_{i-1}\|^2_\H\\
 &\leq 2 \k_1\k_2|\Delta|  q_{i}^{-1} \prod_{j=1}^{i-1} \alpha_j^2.
\end{align*}
 Since $\rho' = \sum_{i=1}^l \rho'_i$ we find
\[
\|\rho'\|_2 \leq \sqrt{2 \k_1 \k_2 |\Delta|} \sum_{i=1}^l q_i^{-1/2}  \prod_{j=1}^{i-1} \alpha_j.
\]
We next choose the parameters $l$, $\alpha_i$ and $\beta_i$ in a suitable way  to show that \ref{eq:dual4}, \ref{eq:dual5} and \ref{eq:dual6} in Proposition~\ref{prop:dual} are satisfied. Letting
\[
l = \left\lceil  \log_2 (\k_1\sqrt{|\Delta|\k_2})+2  \right\rceil,\quad \alpha_1 = \alpha_2 = \frac{1}{4\sqrt{\sqrt{\k_2}\log  (|\Delta|\k_1^2\k_2)}},\quad \beta_1=\beta_2 = \frac{1}{7\sqrt{|\Delta|\k_2}}
\]
and for $i\geq 3$
\begin{align*}
\alpha_i = \frac{1}{2}, \quad \quad \quad \quad \quad\beta_i = \frac{4\log ( |\Delta|\k_1^2\k_2)}{7\sqrt{|\Delta|}},
\end{align*}
from the above estimates we readily derive
\[
\Vert D^\ast\sgn(P_\Delta D g_0)-P_{\W} \rho \Vert_\H \leq \frac{1}{16\k_1\sqrt{\k_2}},\quad \Vert P_\Delta^{\perp} D^{-*} P_{\W}^\perp \rho \Vert_\infty \leq \frac{1}{4},\quad \|\rho' \|_2 \leq \sqrt{|\Delta| \k_1\k_2} \L,
\]
where $\L =\sqrt{2}\sum_{i=1}^l q_{i}^{-1/2} \prod_{j=1}^{i-1} \alpha_j$ will be estimated at the end of the proof. 

Next, we need to establish that the total number of sampled $\Omega_i^j$ remains small with high probability. More precisely, we will bound the probability
$$
p_4 = \P \left( (r_1 > 1) \quad \text{or} \quad (r_2 >1)\quad \text{or} \quad \sum_{i=1}^l r_i > l' \right)
$$
for some $l'$ to be chosen later. To that end, denote  the probability that \eqref{c1} fails in the i-th step by $p_1(i)$,  the probability of failure for \eqref{c2} by $p_2(i)$ and  the probability of failure for \eqref{c3} by $p_3(i)$. We want to use Propositions~\ref{Proposition9.2}, ~\ref{prop9.1} and~\ref{prop:new4} to bound these probabilities. 
Proposition~\ref{Proposition9.2} for $t = \alpha_i$ gives the estimate
\begin{align*}
p_1(i) \leq \exp\left( \frac{-\alpha_i^2 q_i}{64|\Delta|\mu^2 \eta_\Delta^2 \k_1} + \frac{1}{4} \right).
\end{align*}
Thus if
\begin{equation*}
q_i \geq \frac{64 \mu^2 |\Delta| \eta_\Delta^2 \k_1}{\alpha_i^2}(\omega + \log\gamma+\frac 1 4), 
\end{equation*}
then $p_1(i) \leq \frac{1}{\gamma}e^{-\omega}$, where $\gamma>0$ will be chosen later. Similarly,  Proposition~\ref{prop9.1} for $t = \beta_i$ yields 
\begin{align*}
p_2(i) \leq {2}\tilde M({q_i}\beta_i/2)\exp\left(\frac{-\beta_i^2 q_i}{8\k_1\mu^2 B_\Delta (B_\Delta+\eta_\Delta\sqrt{2|\Delta|}\beta_i/6)} \right). 
\end{align*}
Thus if
\begin{align*}
q_i \geq \frac{8 \mu^2\k_1 B_\Delta (B_\Delta+\eta_\Delta \sqrt{2|\Delta|}\beta_i/6)}{\beta_i^2}(\omega + \log({2}\tilde M({q_i}\beta_i/2) \gamma) ),
\end{align*}
then $p_2(i) \leq \frac{1}{\gamma}e^{-\omega}$. Finally,  Proposition~\ref{prop:new4}  yields 
\begin{align*}
p_3(i) \leq \exp\left( \frac{-\k_1 q_i}{16 \mu^2 \eta_\Delta^2 |\Delta|}+\frac 1 4\right). 
\end{align*}
Thus if
\begin{align*}
q_i \geq 16 \mu^2 |\Delta| \eta_\Delta^2(\omega + \log\gamma+\frac 1 4), 
\end{align*}
then $p_3(i) \leq \frac{1}{\gamma}e^{-\omega}$. Choose $\gamma=9$.
Assume $q_i$ are chosen as follows:
\begin{align}\label{cond:q12}
q_1=q_2 &\geq c\mu^2\k_1\k_2 \eta_{\Delta}^2 |\Delta| B_{\Delta}^2 \omega\log \left( |\Delta|\k_1^2\k_2\tilde M({q_1}\beta_1/2)\right),\\ 
\label{cond:q3}
q_i &\ge c\mu^2 \k_1\k_2 \eta_{\Delta}^2 |\Delta| B_{\Delta}^2 \omega\frac{\log \tilde M({q_i}\beta_i/2)}{\log  (|\Delta|\k_1^2\k_2)}, \; \quad i\ge 3,	
\end{align}
where $c\ge 100$ is an absolute constant sufficiently large so that $p_1(i),p_2(i), p_3(i) \leq \frac 1 9 e^{-\omega} \leq {\frac{1}{24}}$ for $i\geq 1$ (we use that $\k_2,B_\Delta,\eta_\Delta\ge 1$, $M\ge 3$ and $ |\Delta|\k_1^2\k_2\ge 3$). In particular, we obtain
\[
\P(\text{\eqref{c1}, \eqref{c2} and \eqref{c3} are satisfied})\ge 7/8, \qquad i\ge 1.
\]
As a consequence, since $\sum_{i=1}^l r_i > l'$ if and only if fewer than $l$ of the first $l'$ samplings satisfied \eqref{c1}, \eqref{c2} and \eqref{c3}, we have
\[
\P \left(\sum\limits_{i=1}^l r_i > l' \right) \le \P (X < l) = \P (X \le l-1), \qquad  X \sim Bin \left(l', {\frac78} \right),
\]
(see equation (45) in \cite{2015-gross-krahmer-kueng}). Thus we need to bound the probability of obtaining less than $l$ outcomes in a binomial process with $l'$ repetitions and individual success probability ${\frac78}$. Following \cite{G,KG} we bound this quantity using a standard concentration bound from \cite{MC}
$$
\P \left( Bin (n, p) - np  \leq -\tau \right) \leq e^{-2 \tau^2 /n} ,
$$
which implies
\begin{align*}
\P \left(\sum\limits_{i=1}^l r_i > l' \right) \leq \exp\left( \frac{-2\left({\frac78}l' - l +1\right)^2}{l'}\right).
\end{align*}
Therefore, choosing {$l' = \frac{16}{7}(l-1)+\frac{32}{49}(\omega+\log 9)$}, we get 
\begin{align*}
\P \left(\sum\limits_{i=1}^l r_i > l' \right) \leq \frac{1}{9}e^{-\omega},
\end{align*}
and, as a consequence, we obtain
\[
 p_4 \le {p_1(1)+p_2(1)+p_3(1)+p_1(2)+p_2(2)+p_3(2)}+\frac{1}{9}e^{-\omega} \le \frac{7}{9}e^{-\omega}.
\]

Let us now consider property \ref{eq:dual1} of Proposition~\ref{prop:dual}. Our aim is to show that
$$
p_{5} = \P \left(\left\Vert\theta^{-1}P_{\W}U^{-1}P_{\Omega}UP_{\W} - P_{\W} \right\Vert > \frac{1}{2} \right) \leq \frac{2}{27}e^{-\omega}.
$$
From Proposition~\ref{Theorem9.3} we immediately obtain that if $\theta$ satisfies 
\begin{equation}\label{cond:theta1}
\theta \geq 70\mu^2 \eta_\Delta^2 |\Delta| \k_1 (\omega+\log |\Delta|+ \log 54),
\end{equation}
then $p_{5} \leq \frac{2}{27} e^{-\omega}$.

Now, let $p_{6}$ be the probability that property \ref{eq:dual3} of Proposition~\ref{prop:dual} fails. We want to show that 
\[
p_{6}=\P\left(\sup_{j \in \Delta^c} \Vert \theta^{-1}P_{\{j\}} D^{-*} P_\W^\perp U^{*}P_{\Omega}U P_{\W}^{\perp} D^{-1}P_{\{j\}} \Vert > 2 \k_1 \k_2 \right) \leq \frac{2}{27}e^{-\omega}.
\]
By Proposition~\ref{prop:unifoff}, if $\theta$ satisfies 
\begin{equation}\label{cond:theta2}
\theta \geq 3 B_\Delta^2 \mu^2  (\omega+\log \tilde M(\theta)  +\log 27),   
\end{equation}
we have $p_{6} \leq \frac{2}{27}  e^{-\omega}$.

Now, let $p_7$ be the probability that property \ref{eq:dual2} of Proposition~\ref{prop:dual} fails, namely
\[
p_7=\P\left(\|\theta^{-1}P_\W U^{-1} P_\Omega U^{-*} P_\W \|>2\k_1\right).
\]
By Proposition~\ref{prop:new},  if $\theta$ satisfies  \eqref{cond:theta1}, then in particular  $p_7\leq \frac{2}{27}  e^{-\omega}$.

In order to finish the proof we need to give a bound on $m$ (or, equivalently, $\theta$) and construct $q_i$ such that conditions \eqref{cond:q12}, \eqref{cond:q3}, \eqref{cond:theta1} and \eqref{cond:theta2} are satisfied.

Let $\theta$ satisfy
\begin{equation}\label{eq:theta}
\theta \geq 9c \mu^2 \eta_\Delta^2|\Delta| \k_1\k_2 \omega^2 B_\Delta^2 \log\left( |\Delta|\k_1^2\k_2 \tilde M \left(  \frac{\beta_1 \theta}{18\omega}  \right)\right).
\end{equation}
Then conditions \eqref{cond:theta1} and \eqref{cond:theta2} are clearly satisfied (using $   \tilde M \bigl(  \frac{\beta_1 \theta}{18\omega}  \bigr) \ge \tilde M(\theta)$ and $9c\ge 900$). Now recall that at each iteration $i$ we sampled $r_i$ sets $\Omega_i^j \sim Ber(q_i)$ and we stopped after $\sum_{i =1}^l r_i \le l'$ sampling. Following the same arguments as in \cite[Section 9.1]{AH}, since
\[
 \Omega=\bigcup_{i=1}^l \bigcup_{j=1}^{r_i} \Omega_i^j,  \qquad \Omega \sim Ber(\theta), \qquad \Omega_i^j \sim Ber(q_i),
\]
we have the identity $\Pi_{i=1}^l (1-q_i)^{r_i} = 1-\theta$, which yields the constraint
\begin{equation}\label{eq:constraint}
\sum_{i=1}^l r_i q_i \geq \theta.
\end{equation}
Define
\[
q_1=q_2 = \frac{\theta}{9},\qquad q =   q_i = 1-\left(\frac{1-\theta}{(1-q_1)(1-q_2)}\right)^{\frac{1}{r_3+\cdots+r_l}},\;i\ge 3.
\]
By \eqref{eq:theta}, condition \eqref{cond:q12} is satisfied (using also $\tilde M \bigl(  \frac{\beta_1 \theta}{18\omega}  \bigr) =  \tilde M(\frac{\beta_1 q_1}{2\omega} ) \geq  \tilde M(\beta_1 q_1/2)$).  By \eqref{eq:constraint}, since $r_1 = r_2 =1$ and $l'\ge \sum_i r_i$, we have
\[
 (l'-2)q \ge\sum_{i=3}^l r_i q_i \ge 
  \theta-2q_1= 
 \frac{7}{9}\theta.
 \]
As a consequence, since  
\[
\begin{split}
l' -2&=  {\frac{16}{7}}\lceil \log_2 \bigl(\k_1\sqrt{|\Delta|\k_2}\bigr) +1 \rceil+{\frac{32}{49}}(\omega+\log 9)-2 \\
 & \le {\frac{16}{7}} \log_2\bigl(\k_1\sqrt{|\Delta|\k_2}\bigr)  +{\frac{32}{7}} +{\frac{32}{49}}(\omega+\log 9)-2 \\ 
  & = {\frac{8}{7}}\log_2 e \log  (|\Delta|\k_1^2\k_2) +{\frac{18}{7}} +{\frac{32}{49}}(\omega+\log 9)\\
  &\leq 7\log(|\Delta|\k_1^2\k_2)\omega,
\end{split}
\]
by using \eqref{eq:theta} it is straightforward to check that condition \eqref{cond:q3} is satisfied as well. Here we have also used the fact that $\beta_i q_i /2  \geq \frac{\beta_1 \theta}{18\omega}  $ for $i \geq 3$.

We can now estimate the constant $\L =\sqrt{2}\sum_{i=1}^l q_{i}^{-\frac12} \prod_{j=1}^{i-1} \alpha_j$. We have:
\begin{equation*}
\begin{split}
\frac{\L}{\sqrt{2}} &= q_1^{-\frac12} + q_2^{- \frac12  }\alpha_1+\sum_{i=3}^l q^{-\frac12}\alpha_1 \alpha_2 \prod_{j=3}^{i-1} \alpha_j\\
& = q_1^{-\frac12} \left(1+\frac{1}{\sqrt{16\sqrt{\k_2} \log(|\Delta|\k_1^2\k_2)}}\right)+\frac{q^{-\frac12}}{16\sqrt{\k_2} \log(|\Delta|\k_1^2\k_2)}\sum_{i=3}^l\frac{1}{2^{i-3}}\\
&\leq \frac54 q_1^{-\frac12}+\frac{q^{-\frac12}}{8 \log(|\Delta|\k_1^2\k_2)} \\
&\leq C_1 \theta^{-\frac12} + C_2\theta^{-\frac12}\frac{(l'-2)^{\frac12}}{ \log(|\Delta|\k_1^2\k_2)} \\
&  \leq C_3\theta^{-\frac12} \omega^{\frac12},
\end{split}
\end{equation*}
where we have used the fact that $|\Delta| \geq 3$, the definition of $q_1,q_2$ and the inequalities above involving $q$, $\theta$ and  $l'-2$ (here $C_1$, $C_2$ and $C_3$ are universal constants).

Finally, the union bound  gives $p_4 + p_5 +p_6+p_7 \leq e^{-\omega}$, which finishes the proof of the proposition (note that the factor $|\Delta|$ in the logarithm may be removed since $\tilde M\bigl( \frac{\beta_1 \theta}{18\omega}  \bigr)\ge M\ge|\Delta|$).
\end{proof}

\subsection{Proof of Theorem~\ref{thm:main}}
The proof is now immediate. By Proposition~\ref{prop:golfing}, under our assumptions with high probability there exists a dual certificate. Thus, by Proposition~\ref{prop:dual} we have
\[
\Vert g - g_0 \Vert_\H \leq 20\k_1\sqrt{\k_2} \|P_{\Delta}^\perp D g_0\|_1+  \epsilon \sqrt{\k_1} \sqrt{\frac{N}{m}}\left( 10+20C'''\k_1 \k_2\sqrt{\omega s} \right)
\]
for every $\Delta\subseteq\{1,\dots,M\}$ such that $|\Delta|=s\ge 3$. Observing that
\[
\begin{split}
\sigma_{s,M}(Dg_0)&=\inf\{\norm{x-Dg_0}_1: \supp(x)\subseteq\{1,\dots,M\},\, |\supp(x)|\le s\}\\
&=\inf\{\norm{x-Dg_0}_1: \supp(x)\subseteq\Delta\subseteq\{1,\dots,M\},\, |\Delta|= s\}\\
&=\inf\{\norm{x-P_\Delta Dg_0}_1+\norm{P_\Delta^\perp Dg_0}_1: \supp(x)\subseteq\Delta\subseteq\{1,\dots,M\},\, |\Delta|= s\}\\
&=\inf\{\norm{P_\Delta^\perp Dg_0}_1: \Delta\subseteq\{1,\dots,M\},\, |\Delta|= s\},
\end{split}
\]
and that
\[
 10+20C'''\k_1 \k_2\sqrt{\omega s} \le C'' \k_1\k_2\sqrt{\omega s} 
\]
for some absolute constant $C''>0$, gives the desired estimate.

\section*{Acknowledgments}
The authors express their sincere thanks to Miren Zubeldia for having largely contributed to this project with many helpful comments and useful insights. They also thank Maximilian M\"arz for suggesting an important simplification in the minimization problem.  M.S.\ would like to thank Eric Bonnetier for stimulating discussions about electrical impedance tomography and compressed sensing that have influenced the present work.  The first conversations about this project took place during the trimester \textit{Program on Inverse Problems}, held at the Institut Henri Poincar\'e of Paris in 2015: it is a pleasure to thank the institute and the organizers of the program for the wonderful opportunity offered. G.S.A.\ acknowledges support from the ETH Z\"urich Postdoctoral Fellowship Program as well as from the Marie Curie Actions for People COFUND Program.

\appendix

\section{On samplings with or without replacement}

Take $m,N,s_1,\dots,s_N\in \N$ with $m\le N$. We select $m$ elements from $\{1,\dots,N\}$ at random (with possible repetitions) following these two distributions.
\begin{enumerate}
\item For $\u\in\N$, let $\Omega'$ be chosen uniformly at random as a subset of cardinality $m$ (namely, without replacement) of
\[
\biggl\{ \underbrace{1, \ldots, 1}_{\u s_1 \text{ times}}, \ldots, \underbrace{N, \ldots, N}_{\u s_N \text{ times}}\biggr\}.
\]
Letting $n'_l$ denote the number of $l$'s in $\Omega'$ for $l=1,\dots,N$, we have
\begin{equation}\label{eq:hyper}
\P(n'_1=k_1,\dots,n'_N=k_N)=\frac{\binom{\u s_1}{k_1}\cdots\binom{\u s_N}{k_N}}{\binom{\u(s_1+\dots+s_N)}{m}},\qquad k_1+\dots+k_N=m,
\end{equation}
according to the multivariate hypergeometric distribution.
\item The set $\Omega$ is formed by selecting $m$ elements from $\{1,\dots,N\}$ with replacement following the variable density sampling
\[
p_l=\frac{s_l}{s_1+\dots+s_N},\qquad l=1,\dots,N.
\]
This corresponds to the multinomial distribution, and we have
\begin{equation}\label{eq:multi}
\P(n_1=k_1,\dots,n_N=k_N)=m!\prod_{l=1}^N \frac{p_l^{k_l}}{k_l!},  \qquad k_1+\dots+k_N=m,
\end{equation}
where $n_l$ denotes the number of $l$'s in $\Omega$ for $l=1,\dots,N$.
\end{enumerate}
A classical result in probability theory states that as the number of objects tends to infinity, the multivariate hypergeometric distribution converges to the multinomial distribution. We need the following quantitative version.
\begin{lemma}\label{lem:ernesto}
Assume that $\u\ge 2m^2$. For every $k_1,\dots,k_N\in\N\cup\{0\}$ such that $k_1+\dots+k_N=m$ we have
\[
\P(n'_1=k_1,\dots,n'_N=k_N)\ge \frac12\,\P(n_1=k_1,\dots,n_N=k_N).
\]
\end{lemma}
\begin{proof}
For $l=1,\dots,N$ we have
\begin{align*}
\binom{\u s_l}{k_l}&=\frac{\u^{k_l}s_l^{k_l}}{k_l!}\left(1-\frac{1}{\u s_l}\right)\cdots \left(1-\frac{k_l-1}{\u s_l}\right),\\
\binom{\u s}{m}&=\frac{\u^{m}s^{m}}{m!}\left(1-\frac{1}{\u s}\right)\cdots \left(1-\frac{m-1}{\u s}\right),
\end{align*}
where $s=s_1+\dots+s_N$.
In view of \eqref{eq:hyper} we obtain
\[
\P(n'_1=k_1,\dots,n'_N=k_N)=\frac{\displaystyle{m!\prod_{l=1}^N\frac{s_l^{k_l}}{k_l!}}\prod_{t=1}^{k_l-1}\left(1-\frac{t}{\u s_l}\right)}{\displaystyle{s^m\prod_{t=1}^{m-1}\left(1-\frac{t}{\u s}\right)}},
\]
where the product $\prod_{t=1}^{k_l-1}$ is equal to $1$ when $k_l=0$ or $k_l=1$.
Thus, identity \eqref{eq:multi} yields
\[
\frac{\P(n'_1=k_1,\dots,n'_N=k_N)}{\P(n_1=k_1,\dots,n_N=k_N)}=\frac{\displaystyle{\prod_{l=1}^N}\prod_{t=1}^{k_l-1}\left(1-\frac{t}{\u s_l}\right)}{\displaystyle{\prod_{t=1}^{m-1}\left(1-\frac{t}{\u s}\right)}}.
\]
Using that the denominator is smaller than $1$ and that
\[
 \prod_{l=1}^N\prod_{t=1}^{k_l-1}\left(1-\frac{t}{\u s_l}\right) \ge  \prod_{l=1}^N\prod_{t=1}^{k_l-1}\left(1-\frac{m}{\u }\right) 
 \ge  \prod_{l=1}^N \left(1-\frac{m}{\u }\right)^{k_l}=\left(1-\frac{m}{\u }\right)^{m},
\]
since $\u\ge 2m^2$ we obtain
\[
\frac{\P(n'_1=k_1,\dots,n'_N=k_N)}{\P(n_1=k_1,\dots,n_N=k_N)} \ge \left(1-\frac{1}{2m }\right)^{m}\ge \frac12,
\]
as desired.
\end{proof}

\bibliography{CS}
\bibliographystyle{plain}

\end{document}